\RequirePackage[2018-12-01]{latexrelease}
\documentclass{rQUF2e}

\usepackage{amsmath}
\usepackage{amsthm}
\usepackage{amssymb}
\usepackage{amsfonts}
\usepackage{bm}
\usepackage{bbm}
\usepackage{color}
\usepackage{graphicx}
\usepackage{hyperref}
\usepackage{fancyhdr}
\usepackage{graphicx}
\usepackage{slashed}
\usepackage{latexsym,epsfig}
\usepackage{algorithm,algpseudocode}	
\usepackage{enumerate}
\usepackage{mathtools}
\usepackage{physics}
\usepackage[capitalize]{cleveref}

\definecolor{blue}{rgb}{0,0.2,1}
\definecolor{Pr}{rgb}{0.4,0.3,0.9}

\newcommand{\Ord}[1]{\mathcal{O}\left( #1 \right)}
\newcommand{\tOrd}[1]{\widetilde{\mathcal{O}}\left( #1 \right)}

\newcommand{\Var}{\mathrm{Var}}
\newcommand\super[1]{^{(#1)}}

\theoremstyle{plain}

\newtheorem{problem}{Problem}
\newtheorem{theorem}{Theorem}

\newtheorem{lemma}{Lemma}

\newtheorem{defn}{Definition}
\newtheorem{cor}{Corollary}
\newtheorem{fact}{Fact}

\def\be{\begin{eqnarray}}
\def\ee{\end{eqnarray}}

\DeclareMathOperator{\EE}{\mathbb{E}}

\def\calB{\mathcal{B}}

\def\RR{\mathbb{R}}

\def\QA{\mathsf{QA}}
\def\QS{\mathsf{QS}}
\def\VA{\mathsf{VA}}
\def\SA{\mathsf{SA}}
\def\QC{\mathsf{QC}}
\def\TV{\mathrm{TV}}
\def\CVA{\mathrm{CVA}}

\DeclareMathOperator{\id}{\mathbbm{1}}
\DeclareMathOperator{\lvec}{vec}
\DeclarePairedDelimiter{\ceil}{\lceil}{\rceil}

\DeclareMathOperator*{\argmin}{arg\,min}

\begin{document}
\title{Quantum Advantage for Multi-option Portfolio Pricing and Valuation Adjustments}

\author{
Jeong Yu Han$^{\ast 1}$\footnote{$^\ast$ jeong.yuhan@gmail.com},
Bin Cheng$^{\dag 2}$ \thanks{$^\dag$ bincheng@nus.edu.sg},
Dinh-Long Vu$^{\ddag 2}$ \thanks{$^\ddag$ long-vu@nus.edu.sg},
and Patrick Rebentrost$^{\S 1, 2}$ \thanks{$^\S$ cqtfpr@nus.edu.sg} \\
\affil{$^1$ Department of Computer Science, National University of Singapore, Singapore \\
$^2$ Centre for Quantum Technologies, National University of Singapore, Singapore, Singapore} }



\maketitle

\begin{abstract}
A critical problem in the financial world deals with the management of risk, from regulatory risk to portfolio risk. Many such problems involve the analysis of securities modelled by complex dynamics that cannot be captured analytically, and hence rely on numerical techniques that simulate the stochastic nature of the underlying variables. These techniques may be computationally difficult or demanding. Hence, improving these methods offers a variety of opportunities for quantum algorithms. 
In this work, we study the problem of Credit Valuation Adjustments (CVAs) which has significant importance in the valuation of derivative portfolios. As a variant, we also consider the problem of pricing a portfolio of many different financial options.
We propose quantum algorithms that accelerate statistical sampling processes to approximate the price of the multi-option portfolio and the CVA under different measures of dispersion.
Technically, our algorithms are based on enhancing the quantum Monte Carlo (QMC) algorithms by Montanaro with an unbiased version of quantum amplitude estimation.
We analyse the conditions under which we may employ these techniques and demonstrate the application of QMC techniques on CVA approximation when particular bounds for the variance of CVA are known.
\end{abstract}


\section{Introduction} \label{intro}
\subsection{Quantitative Finance}
Derivatives are financial securities whose values are derived from underlying asset(s).   
The derivatives market has seen a rapid expansion and is estimated to be up to more than ten times of the global Gross Domestic Product \cite{Stankovska2017}. 
Financial theory has co-evolved with this expansion. 
A derivative contract typically consists of payoff functions that are often dependent upon market state variables realised in the future, which are necessarily random.  Study of these random variables and accurately pricing these financial contracts is the principle task of quantitative finance, which lends itself to theoretical treatment using measure and probability theory, stochastic processes, and numerical methods. In particular,  the evolution of the underlying assets are modelled as Ito processes \cite{Shreve2004}, and may be analysed under the settings of stochastic calculus.  The evolution of the price process then can be modelled using stochastic differential equations (SDEs).

Unfortunately,  analytical solutions are often not known or are prohibitively complex to formulate,  especially when one considers the interactions between a basket of derivative contracts.  Such complications are the main reasons why Monte Carlo engines have become an integral part of financial modelling,  providing a general,  numerical approach to obtain solutions to compute expectations of random variables even in high dimensional settings.

In cases when an analytical solution is possible, derivative products may be valued by the distinguished Black-Scholes-Merton (BSM) \cite{Black1973} formula. 
However, this formula makes a set of well-known but limiting assumptions.
Of primary concern to the dealing house involved in the transaction of financial derivatives may be outlined; to what extent should the valuation of derivative portfolios go beyond the BSM model based on the idiosyncratic characteristics of the parties in question? Such a question is answered in reality by the practice of XVAs, where the `X' refer to a number of items such as Credit (C), Debt (D),  Funding (F),  Margin (M), Capital Valuation (K) and the `VA' referring to Valuation Adjustments \cite{Green2015}. Each term pertains to one of credit,  funding and regulatory capital requirements that modifies the adjusted derivative portfolio value compared to the BSM price. Further elaborations on the CVA problem and the role of XVA desks are presented in \cref{append_CVA}.

Our examination in particular concerns the credit component, also known as CVA.  While the collapse of Lehman Brothers 
\cite{Lehman2013} is likely the most ill-famed credit event in popular culture,  credit events are anything but a rare occurrence in finance.  The credit crisis of 2007 and collapse of Lehman Brothers in 2008 brought to attention systemic risks in the financial markets and the need for better modelling of risks. In response, regulatory frameworks have been developed by the Basel Committee on Banking Supervision to mitigate further risks of financial crisis.
Despite stricter controls and updates to the theory of derivative pricing, the market still differs significantly in pricing practice \cite{Zeitsch2017} with XVA desks applying varying levels of adjustments. This is attributed to the complexity of the modelling and controversial modelling standards, resulting in divergent prices and two-tier markets. Failure to accurately model these risks have resulted in large losses. The blowup of Archegos Capital in March 2021 is an example of institutional failure to accurately price counterparty credit risk, and the losses incurred by the likes of Credit Suisse and Nomura Holdings evidence the consequences \cite{archegos}. It follows that our subject of interest is of critical importance in both methodology and implementation therein.

\subsection{Quantum algorithms}
Quantum computing exploits quantum mechanical phenomena such as superposition and entanglement to perform computation on quantum states formed by quantum bits (qubits).  In some cases,  quantum algorithms promise speedups over their classical counterparts. An example is the integer factoring problem, where the Shor's algorithm \cite{Shor1999} allows an exponential speedup over classical algorithms in the factorization of integers. Another useful algorithm in search problems under the quantum setting is known as the Grover's search, \cite{Grover1996}, that allows the search for specific data in an unstructured table using $\Ord {\sqrt N}$ queries, promising quadratic speedups to the best classical counterpart.

Another well-known algorithm that proves particularly useful is Quantum Amplitude Estimation (QAE) \cite{brassard2002quantum}, which estimates the amplitude of an arbitrary quantum state $\ket{\psi}$ within a subspace of the state space. 
For example, advancements in Quantum Monte Carlo (QMC) techniques \cite{Montanaro2015} have been achieved by generalizing QAE and other useful algorithms \cite{durr1996} as subroutines.
Several variants of QAE have been developed~\cite{grinko_iterative_2021,giurgica-tiron_low_2022,cornelissen_sublinear-time_2023,rall_amplitude_2023}, with unbiasedness emerging as a prevalent concept.
Unbiased QAE turns out to be a useful subroutine for estimating partition functions~\cite{cornelissen_sublinear-time_2023} and constructing low-depth algorithms~\cite{vu_low_2024}.

Several works consider the applications of quantum computing in finance. A general review of today's challenges in quantum finance can be found \cite{bouland2020prospects}, where additional literature such as quantum linear algebra and quantum machine learning are discussed.

\subsection{Monte Carlo methods for CVA computations}

Monte Carlo is a statistical sampling method used to estimate properties of statistical distributions that are difficult to estimate analytically.  Consider a non-empty set of random variables $S$ each from some, not necessarily identical, distribution.  Let $g$ be a Borel-measurable function over $S$ on $\mathbb{R}$. Our objective is to find the mean value $\mu = \EE \left[g(S)\right]$.  We can calculate $\bar{x}$, the sample average as an estimate of $\mu$ via statistical sampling. Assuming that the variance is bounded by $\sigma^2$, Chebyshev's inequality guarantees an upper bound on the probability of the accuracy of our estimate $\bar{x}$ up to error $\epsilon$, such that $\mathbbm P\left[ |\bar{x} - \mu| \geq \epsilon \right] \leq \frac{\sigma ^2}{n\epsilon^2}$. The statistical exhibit of quadratic dependence on $\frac{\sigma}{\epsilon}$ is less than desirable, since the number of samples required to obtain an estimate up to additive error of $\epsilon$ is $n \in \Ord{\frac{\sigma^2}{\epsilon^2}}$.

Notably, Montanaro's work demonstrates the near quadratic speedup over the best classical methods in the estimation of mean output values of arbitrary randomized algorithms in the general settings for QMC~\cite{Montanaro2015}.
Since then, new research has found improvements in the quantum advantage for multilevel Monte Carlo methods for SDEs found in mathematical finance \cite{An2020}, as well as for quantum multi-variate Monte Carlo problems \cite{Arjan2021}.
But the QMC algorithms developed in \cite{Montanaro2015} are based on the convectional QAE~\cite{brassard2002quantum} and therefore, is not guaranteed to output unbiased estimators.
For some applications, it is desirable to consider unbiased QMC algorithms.

Classical algorithms to approximate CVA computations also use the Monte Carlo engines~\cite{Green2015}.
But until recently, little to no known literature has expounded on quantising the CVA computation. A recent work demonstrated the first attempt of quantising the CVA formula, introducing numerous heuristics and the adoption of QAE variants to reduce the circuit depth and resource requirements for implementations on near-term quantum devices~\cite{Zapata2021}. 
In particular, a Bayesian variant using engineered likelihood functions was explored, while using standard techniques for accelerating Monte Carlo sampling techniques by Montanaro. 
However, in the fault-tolerant setting, it remains unexplored the theoretically provable quantum speedup from applying QMC to CVA.

\subsection{Our results}
In this work, we study the problems of multi-option pricing and CVA, which can be formalised as mean value estimation problems. 
We show that CVA can be viewed as a version of the multi-option pricing problem.
We then develop quantum algorithms for multi-option pricing and CVA using quantum Monte Carlo.
Using the unbiased quantum amplitude estimation algorithm~\cite{cornelissen_sublinear-time_2023}, we rework and enhance the quantum Monte Carlo algorithms in Ref.~\cite{Montanaro2015} under several different settings, which might be of independent interest.

We find that under guarantees of a bounded variance in the CVA, we may provide better results for approximating the CVA value than the general settings \cite{Zapata2021}. We find that such guarantees are reasonably common and useful in practical settings. In particular, we may obtain an approximation of the CVA value up to desired additive error $\epsilon$ using $\tOrd{\frac{\sigma (1 - R)}{\epsilon} \log{\frac{1}{\delta}}} $ queries with success probability $1 - \delta$, where $R$ is the recovery rate (\cref{theorem_additive}).

Additionally, if we are given a variance bound such that $\Var(\CVA) \leq B \cdot \mathbb{E}[{\CVA}]^2$ for some $B > 0$, then we may obtain an approximation of the CVA value up to desired relative error $\epsilon$ using  $\tOrd{\frac{B}{\epsilon} \log{\frac{1}{\delta}}}$ queries with success probability $1 - \delta$. See \cref{theorem_relative} for a precise statement of this result. We compare these algorithms to the setting where no variance guarantees can be obtained.

\subsection{Organization of this work}
In \cref{preliminaries}, we discuss preliminaries and notations required for formalising and introducing our problem settings. 
We provide formal definitions for the multi-option pricing problem and the CVA problem in \cref{problem_statements}, where we see that the CVA problem can be framed as a variant of the multi-option pricing problem. We discuss quantum subroutines and relevant theory in \cref{subroutines}. In \cref{solutions}, we argue that the solutions to the problem statements may be found using these algorithms. We conclude and provide some guidance for possible future work in \cref{conclusions}. In the \cref{append_CVA}, we discuss more on CVAs and provide the derivation of the CVA formula. We discuss relevant quantitative finance topics of obtaining default probabilities and discount factors in \cref{append_quant}.

\section{Preliminaries} \label{preliminaries}
In this section, we discuss notations and definitions.

\subsection{Mathematical preliminaries}
The following are mathematical notations and definitions used in discussions throughout the paper.
We denote by $\Vert \vb{v} \Vert_1 := \sum_{i=1}^n |v_i|$ the $\ell_1$-norm of the vector. 
The simplex of non-negative, $\ell_1$-normalized $N$-dimensional vectors are defined as $\Delta^N := \left \{ \vb{u} \in [0, +\infty)^{N} : \sum_{k=1}^N u_k = 1\right \}$.
We denote by $\vb{u} \cdot \vb{v} = \sum_{i=1}^N u_i v_i$ the inner product between two vectors.
For any matrix $A = (\vb{a}_1, \cdots, \vb{a}_m)$ with columns $\vb{a}_j, \forall j\in[m]$, define the vector constructed by stacking the columns of $A$ as ${\rm vec}(A)$ as the vectorization of a matrix.

In our work, when working with numbers in the classical setting, we assume an arithmetic model with negligible encoding errors. Additionally, arithmetic operations all cost $\Ord{1}$ time. When working with operations in the quantum setting, we use the fixed-point encoding as defined below, and also assume an arithmetic model with negligible encoding errors and $\Ord{1}$ time.

\begin{defn}[Notation for fixed-point encoding of real numbers] \label{encoding}
Let $c_1,c_2$ be positive integers and let $r \in [0, 2^{c_1}]$. 
There exist $\vb{a} \in \{0, 1\}^{c_1}$ and $\vb{b} \in \{0, 1\}^{c_2}$ such that $|r - a_{c_1} \cdots a_1.b_1 \cdots b_{c_2}| \leq 1/2^{c_2}$, where $a_{c_1} \cdots a_1.b_1 \cdots b_{c_2} = 2^{c_1-1}a_{c_1} + \cdots + 2 a_2 + a_1+ \frac{1}{2} b_1 + \cdots + \frac{1}{2^{c_2}} b_{c_2}$.
Define the binary encoding of $r$ to precision $2^{-c_2}$ to be $\calB(r, c_1, c_2) := (\vb{a}, \vb{b})$.
\end{defn}

Note that negative numbers can also be represented with fixed-point encoding. 
It suffices to add an extra bit to represent the sign.
Next, we define the access to elements of a vector in the classical setting under the query access model and the sampling access model.
\begin{defn} [Vector access]\label{defCA}
Let $c_1, c_2$ and $n$ be positive integers
and $\mathbf{u} \in [0, 2^{c_1}]^n$ be a vector. 
We say that we have access to a vector $\vb{u}$ if we have access to the mapping $j \to \calB(u_j, c_1, c_2)$.
We denote this access by $\VA(\vb{u},n,c_1, c_2)$ and the time for a query by $\mathcal T_{\VA(\vb{u}, n,c_1, c_2)}$. When there is no ambiguity of the inputs, we use the shorthand notation $\VA(\vb{u})$.
\end{defn}

\begin{defn} [Sampling access]\label{defSA}
Let $\vb{v} \in \RR^n$ be a vector. 
We say we have sampling access to $\vb{v}$  if we can draw a sample $j\in [n]$ with probability $|v_j|/\Vert \vb{v} \Vert_1$. 
We denote this access by $\SA(\vb{v})$.
\end{defn}

We make use of quantum subroutines in our work to perform sampling on vector elements for inner product estimations. The classical analogue is presented in the following \cref{lemmaSamplingl1}, which has been adapted from \cite{Tang2018} and written as in \cite{rebentrost2021alphatron}.
\begin{lemma}[Inner product with $\ell_1$-sampling]
	\label{lemmaSamplingl1}
	Let $\epsilon,\delta \in (0,1)$. Given vector access to $\vb{v} \in\mathbbm R^n$ and
	sampling access to $\vb{u} \in\mathbbm R^n$,
	we can determine $ \vb{u}\cdot \vb{v}$ to additive error $\epsilon$ with success probability at least
	$1-\delta$ with $\Ord{\frac{\Vert \vb{u} \Vert_{1}^2 \Vert \vb{v} \Vert_{\max}^2}{\epsilon^2} \log \frac{1}{\delta}}$
	queries and samples.
\end{lemma}
\begin{proof}
Define a random variable $Z$ with outcome ${\rm sgn}(u_j) \Vert \vb{u}\Vert_1 v_j$ with probability $\vert u_j\vert/\Vert \vb{u}\Vert_1$.
Note that $\mathbbm E[Z] = \sum_j {\rm sgn}(u_j) \Vert \vb{u}\Vert_1 v_j \vert u_j\vert /\Vert \vb{u}\Vert_1 = \vb{u} \cdot \vb{v}$.
Also, $$\Var[Z] \leq \mathbbm E[Z^2] = \sum_j \Vert \vb{u}\Vert_1^2 v_j^2 \vert u_j\vert/\Vert \vb{u}\Vert_1 \leq
\Vert \vb{u}\Vert_1^2 \Vert  \vb{v} \Vert_{\max}^2.$$
Take the median of $6 \log 1/\delta$ evaluations of the mean of
$9\Vert \vb{u}\Vert_1^2 \Vert  \vb{v} \Vert_{\max}^2/(2\epsilon^2)$ samples of $Z$. Then, by using the Chebyshev and Chernoff inequalities, we obtain an $\epsilon$ additive error estimation of $\vb{u} \cdot \vb{v}$
with probability at least $1-\delta$ in $\Ord{\frac{\Vert \vb{u}\Vert_1^2 \Vert  \vb{v} \Vert_{\max}^2}{\epsilon^2}\log \frac{1}{\delta}}$ queries. 
\end{proof}

\subsection{Probability preliminaries}
To aid in the formal treatment of derivative pricing and CVA concepts, we introduce relevant concepts in probability theory required to construct arguments on asset price dynamics.

\emph{Probability Spaces and Filtrations} --- We denote an arbitrary probability space by $(\Omega,  \mathcal F,  \mathbbm Q)$ where $\Omega$ is a non-empty sample space, $\mathcal F$ the filtration and $\mathbbm Q$ the probability measure.  Consider for some fixed, positive integer $T$, we let $\mathcal F_t, 0 \leq t \leq T$ be a filtration of sub-$\sigma$ algebras of $\mathcal F$. We may take $\mathcal F_t$ then to be market state variables and information available up to time $t$. 
 
\emph{Stochastic Processes} --- Consider a collection of random variables $V_t$  and a collection of $\sigma$-algebras $\mathcal{F}_t,  0 \leq t \leq T$, for which $\mathcal F_0 \subseteq \mathcal F_1 \subseteq \cdots \subseteq \mathcal F_{T-1} \subseteq \mathcal F_T$, over a non-empty sample space $\Omega$ and fixed positive integer $T$. The collection of random variables denoted by $V$ and indexed by $t$ is an adapted stochastic process if $\forall t \in [T]$, $V_t$ is $\mathcal{F}_t$ measurable. 
 
\emph{Probability Measures} --- Consider two random variables $X, Y : \Omega \to \mathbb{R}$ taking some real values. Let $\mathbbm Q$ be the joint distribution measure for $(X, Y)$, then we have $\mathbbm Q((X, Y) \in \mathbbm C)$ defined for all Borel sets $\mathbbm C \subset \mathbbm R^2$.  Further assume that $X$ and $Y$ are independent. Then, their distribution measure factors into $\mathbbm Q_X\left(X \in \mathbbm C\right)$ and $\mathbbm Q_Y\left(Y \in \mathbbm C\right)$.

\subsection{Financial preliminaries}
\label{subsec:financial_preliminaries}

The derivative pricing and CVA pricing problems relate different variables pertaining to the contractual agreement and market state variables; here we discuss some of these variables involved in such computations.

\emph{Portfolio Process} --- Assume a probability space $(\Omega, \mathcal F, \mathbbm Q)$, where  $\Omega$ is the set of economic events,  $\mathcal F$ is the sigma algebra for $\Omega$, and a probability measure $\mathbbm Q$.  Let $T$ be a fixed, positive integer denoting the number of time steps in the model economy. We let $\mathcal F_t, 0 \leq t \leq T$ be a filtration of sub-$\sigma$ algebras of $\mathcal F$, where $\mathcal F_0 \subseteq \mathcal F_1 \subseteq \cdots \subseteq \mathcal F_T =\mathcal F$. Define the discounted portfolio value (of a single derivative or a basket of derivatives) to be a random variable $V: \Omega \to [0,\infty)$.  Let $V_t$ be a $\mathcal F_t$ measurable adapted stochastic process, representing the discounted portfolio value at time $t$,  $\forall t \in [T]$.  

There are numerous events that can lead to a credit event, and of the most common nature may be attributed to operations of financially unsound nature. We formally define the concepts of random variables on default times, survival probabilities and the recovery rates.

\emph{Credit/Default Event} --- Let $\tau : \Omega \to [T] \cup \{\infty\}$ be the random variable for the time of a credit event (e.g., a  bankruptcy). The list of credit events include but are not limited to those outlined by the International Swaps and Derivatives Association \cite{isda2003}. Further discussions on Financial Law is not discussed as it is not central to our work.   

\emph{Default Probabilities} --- Let $\phi : \{t : t \in [T] \cup \{\infty\} \} \to [0,1)$ be the cumulative distribution function for credit default, such that $\phi(\tau < t)$ is the probability that some counterparty in concern defaults at time prior to $t$.  $\phi$ shall be defined in the range of $\tau$, the time of a credit event. Without ambiguity, define $\phi(t_i \leq \tau \leq t_j) := \phi(t_j) - \phi(t_i)$ to be the probability of default between two time instances over the domain $[T]$, where $t_i \leq t_j$. The default probabilities implied by the Credit Default Swap (CDS) market may be obtained by bootstrapping hazard rates under the risk-neutral measure. For a more detailed explanation on deriving default probabilities from the CDS curve, refer to \cref{cds_bootsrapping}. 

\emph{Recovery Rate} --- Let $R \in (0,1)$ be the Recovery Rate,  a percentage of the value of the portfolio that may be expected to be recovered in the event of default of the counterparty. The percentage $(1 - R)$ can then be defined as the Loss Given Default (LGD) rate, representing as percentage of the positive exposure subject to loss under default.
 
\emph{Discounting} --- Let $r_t \in (0,1)$ be the short rate at time $t \geq 0$. Then the discount process is $\exp\left\{ - \int^t_0 r_u du \right \}$ and for future valuations of portfolio $V_t$, the discounted portfolio value today is written $V_0 = \mathbb E\left[ \exp \{ -\int^t_0 r_u du \} V_t \right]$. Short rates may be defined to be deterministic or stochastic. Short rate models may be calibrated to the yield curve. For a more detailed explanation, refer to \cref{irterms}.

\subsection{Quantum preliminaries}

\emph{Quantum query (multiple) access} --- We define quantum query access for obtaining the superposition over elements of a vector.
\begin{defn} [Quantum query access]\label{defQA}
Let $c_1, c_2$ and $n$ be positive integers
and $\vb{u} \in [0, 2^{c_1}]^n$ be a vector. 
We say that we have  quantum access to $\vb{u}$ if, for arbitrary $b \in \{0,1\}^{c_1+c_2}$, 
\be
\ket j \ket{b} \to \ket j \ket{b \oplus \calB(u_j, c_1, c_2)}. 
\ee
We denote this access by $\QA(\vb{u}, n,c_1, c_2)$. Denote the time for a query by $\mathcal T_{\QA(\vb{u}, n,c_1, c_2)}$. 
When there is no ambiguity of the inputs, we use the shorthand notations $\QA(\vb{u})$ and 
\begin{align}
    \ket{u_j} := \ket{\calB(u_j, c_1, c_2)} \ .
\end{align}
The quantum oracle query on a superposition follows directly $
\sum^n_{j=1} \ket j \ket{0^{c_1+c_2}} \to \sum^n_{j=1} \ket j \ket{ u_j}$.
\end{defn}

\begin{defn} [Quantum matrix access/Quantum multi-vector access]\label{defQM}
Let $c$, $n$, and $m$ be positive integers
and $\vb{u}_1,\cdots,\vb{u}_m$ be
$m$ vectors of bit strings $\in [0, 2^c]^n$. 
We say that we have  quantum access to the matrix $\vb{A}:= (\vb{u}_1,\cdots, \vb{u}_m)$ if we have access to $\QA({\rm vec} (\vb{A}))$.
Note that we can interpret this access as a superposition access to the set of inputs  $\QA(\vb{u}_1),$ $\cdots,$ $\QA(\vb{u}_m)$. It allows the operation $\ket i \ket j \ket{0^c} \to \ket i \ket j \ket{(\vb{u}_i)_j}$, for $i\in[m]$ and $j \in[n]$.
\end{defn}

\emph{Quantum sampling (multiple) access} --- We define quantum sampling access into a superposition of basis states, where the probability of observing $j$ under a measurement corresponds to the square of its amplitude.

\begin{defn} [Quantum sample access]\label{defQS}
Let $c$ and $n$ be two positive integers
and $\vb{v} \in [0, 2^{c}]^n$ be a vector. 
Define quantum sample access to $\vb{v}$ via the operation 
\begin{align}
\ket{\bar 0} \to \frac{1}{\sqrt{\Vert \vb{v} \Vert_1}} \sum_{j=1}^n  \sqrt{v_j} \ket j.
\end{align}
We denote this access by $\QS(\vb{v})$. Denote the time for a query by $\mathcal T_{\QS(\vb{v})}$. 
\end{defn}

\begin{defn} [Quantum multi-sample access]\label{defMS}
Let $c$, $n$, and $T$ be positive integers
and $\vb{v}_1,\cdots,\vb{ v}_T \in [0, 2^c]^n$ be vectors.
Define quantum multi-sample access to $\vb{V}:=( \vb{v}_1, \cdots, \vb{v}_T)$ via the operation
\begin{align}
\ket{\bar 0} \to  \sum_{i=1}^T \sum_{j=1}^n  \sqrt {\frac{ (\vb{v}_i)_j}{\Vert \mathrm{vec}(\vb{V}) \Vert_1}} \ket i \ket j.
\end{align}
We denote this access by $\QS(\{\vb{v}_i\}_{i\in[T]})$. Denote the time for a query by $\mathcal T_{\QS(\{\vb{v}_i\}_{i\in[T]})}$. 
\end{defn}

We note that this is just an instance of the \cref{defQS}. Consider the vectorization $\mathrm{vec}(\vb{V})$, the column vector of dimension $nT$. Now consider the $\QS({\rm vec}(\vb{V}))$, which by definition provides the access:
\begin{align}
\ket{\bar 0} &\to \frac{1}{\sqrt{\Vert \mathrm{vec}(\vb{V}) \Vert_1}} \sum_{k=1}^{nT} \sqrt{\mathrm{vec}(\vb{V})_k} \ket k \\ 
&= \frac{1}{\sqrt{\sum_{k=1}^{nT} \vert  \mathrm{vec}(\vb{V})_k } \vert} \sum_{i=1}^T \sum_{j=1}^n  \sqrt{ v_{ij} } \ket i \ket j \\ 
&= \frac{1}{\sqrt{\sum^T_{i= 1} \sum^n_{j = 1} \vert v_{ij} \vert}} \sum_{i=1}^T \sum_{j=1}^n  \sqrt{ v_{ij} } \ket i \ket j \ ,
\end{align}
where $v_{ij} := (\vb{v}_i)_j$.

We note the fact that any classical circuit can be implemented by an equivalent, reversible quantum circuit of unitary mappings.

\begin{fact} [Reversible Logic Synthesis \cite{nielsen2002quantum}] \label{rls}
Given any classical arithmetic computation implemented by $\mathcal T_{cl}$ gates, we may implement an equivalent quantum circuit using $\tOrd{\mathcal T_{cl}}$ gates.
\end{fact}

For instance, for an arbitrary vector $\vb{v} \in \mathbb{R}^n$, the classical operation $\max(\vb{v}, 0)$ can be implemented using quantum circuits.
\begin{lemma}
\label{quantum decomposition}
    Consider a vector $\vb{v} \in \mathbb{R}^n$.
    For any $j \in [n]$, the entry $v_j$ may be represented up to desired accuracy using fixed point binary encoding as in \cref{encoding}.  
    Then, $\forall j \in [n]$, $v_j$ may be decomposed into a difference between two non-negative components such that $v_j = (v_j)^+ - (v_j)^-$ represents positive and negative values.  Assume quantum oracle access $\QA(\vb{v})$.  We may obtain $\QA((\vb{v})^+)$ using two queries to $\QA(\vb{v})$ and additional quantum circuits of depth $\Ord{c}$, where $c$ is the number of bits in the binary encoding of $\vb{v}$. 
\end{lemma}

\begin{proof}
Consider the element-wise operations
\be
\ket j \ket{\bar 0} \ket{\bar 0} \to \ket j \ket{v_j} \ket{\bar 0} &\to&
\begin{dcases}
	\ket j \ket{v_j}\ket{v_j}, & \text{if } v_j \geq 0\\
    \ket j \ket{v_j}\ket{\bar 0},  & \text{otherwise}
\end{dcases}\\
&\to& \begin{dcases}
	\ket j \ket{\bar 0}\ket{v_j}, & \text{if } v_j \geq 0\\
    \ket j \ket{\bar 0}\ket{\bar 0},  & \text{otherwise}
\end{dcases}\\
\ee
This achieves
\be
 \ket j \ket{\bar 0} 
&\to& \begin{dcases}
	\ket j \ket{v_j}, & \text{if } v_j \geq 0\\
   \ket j \ket{\bar 0},  & \text{otherwise}
\end{dcases}\\
\ee
which is by definition $\QA(\max(\vb{v}, 0))$, i.e., $\QA((\vb{v})^+)$.
\end{proof}

Similarly, for arbitrary $x \in \mathbb{R}$ and $S \subseteq \mathbb{R}$, we may implement the classical operation $f(x) = x \cdot {\mathbbm{1}\{x \in S\}}$ using quantum circuits.
\begin{defn} [Quantum comparators] \label{defQC}
Let $c, n$ be positive integers,  $l$ be a non-negative integer
and $\vb{u} \in [0, 2^c]^n$ be a vector. 
Define element-wise bounded quantum access to $\vb u$ for $j\in[n]$ by the operation 
\be
\ket*{\tilde{l}} \ket{u_j}\ket{\bar 0} \to
\begin{dcases}
\ket*{\tilde l}\ket{u_j}\ket{u_j},& \text{if } l = 0,  0 \leq u_j < 1\\
    \ket*{\tilde l}\ket{u_j}\ket{u_j},& \text{if } l > 0, 2^{l - 1} \leq u_j < 2^l\\
    \ket*{\tilde l}\ket{u_j}\ket{\bar 0},              & \text{otherwise}
\end{dcases}
\ee
on $\Ord{l + c}$ qubits, where $\tilde l$ is the bit string representation of $2^l$.
We denote this access by $\QC(\vb{u}, n,c, l)$.  Denote the time for a query by $\mathcal T_{\QC(\vb{u}, n,c, l)}$. When there is no ambiguity of the inputs, we use the shorthand notation $\QC(\vb{u}, l)$.
\end{defn}

\emph{Quantum controlled rotations} --- We define quantum controlled rotations of bounded input states into amplitudes.
\begin{defn} [Quantum controlled rotation]\label{defQCR}
Let $u \in [0 ,1]$ be a number with fixed-point encoding of $c$ bits. 
Define quantum controlled rotation as the operation 
\begin{align}
\ket{u}\ket{0} \to \left(\sqrt{1 - u} \ket 0 + \sqrt{u} \ket 1\right)
\end{align}

The cost of this operation depends directly on the precision of the fixed point arithmetic model (see \cref{encoding}) used. In particular, we neglect the cost of $\Ord c$ in our computational model and assume this to be of unit cost in the following discussions.

\end{defn}

\section{Problem Statements} \label{problem_statements}
In this section we formalise the multi-option pricing problem and the CVA problem.

\subsection {Problem statements for multi-option pricing} \label{option_problem}
In this section we introduce the pricing problem for the general case of a basket of derivatives, and formalise the classical and quantum contexts.

A fairly general classical multi-option pricing problem may be phrased as follows. We have a probability space $(\Omega, \Sigma, \mathbbm Q)$, where $\Omega$ is the set of economic events,  $\Sigma$ is the sigma algebra for $\Omega$, and a probability measure $\mathbbm Q$. 
We are given a portfolio of $K$ options or other financial derivatives.
For each option, we have a discounted payoff $V\super k: \Omega \to [0,\infty)$. 
The price of each option is computed by
$\mathbbm E_{\mathbbm Q}\left[V\super k\right]$.
The problem is to determine the total portfolio value
$\sum_{k=1}^K\mathbbm E_{\mathbbm Q}\left[V\super k\right]$.
In this work, we focus on a more specialized problem. 
We are given a portfolio of $K$ options or other financial derivatives, which we can price independently.
\begin{problem}[Classical multi-option pricing problem under independent,  finite settings]\label{classical_independent_option_pricing} 
Let $K$ be a positive integer for the number of financial derivatives. Assume there exists a known integer $n$, such that for each option indexed by $k\in [K]$ we have a probability space $(\Omega\super k, \Sigma\super k, \mathbbm Q\super k)$, where  $\Omega\super k = \{ x : x\in \{0,1\}^n \}$ describes the set of economic events,  
$\Sigma^{(k)}$ is the sigma algebra for $\Omega\super k$, and $\mathbbm Q\super k$ is the probability measure.
The probability measures are given via $\VA(q\super k)$ access to the vectors $q\super k  \in [0,1]^N$ for which $N=2^n$ and $\sum_{j=1}^N q_j\super k = 1$.
Each option is defined via a discounted payoff $V\super k: \Omega\super k \to [0,\infty)$, for which we are given the vector access $\VA\left(V\super k\right)$.
The price of each option shall be computed by
\be
\mathbbm E_{q\super k}\left[V\super k\right] = \sum_{j=1}^N q\super k_j V\super k_j.
\ee
Define the random variable of the total portfolio value
$
\TV := \sum_{k=1}^K V\super k$.
The task is to evaluate 
\begin{align}
\mathbbm E_{q \super 1, \cdots, q \super K }\left[\TV\right] = \sum_{k=1}^K\mathbbm E_{q\super k}\left[V\super k\right].
\end{align}
\end{problem}

In \cref{classical_independent_option_pricing}, when we access a probability value, we select a certain $k\in[K]$ and then a certain index $j\in[N]$ to obtain $q\super k_j$. 
The natural quantum extension of this process is to be able to query both the index $k$ and the index $j$ in superposition. This ability is embodied in the next definition of the quantum version of the same problem.

\begin{problem}[Quantum multi-option pricing problem under independent,  finite settings] \label{problemCountableQuantum}
For $j \in [N]$ and $k \in [K]$, let $V\super k_j \in [0,\infty)$ be a number that can be represented by at most $c$ bits using the fixed point encoding as in \cref{encoding}. 
Given the setting in \cref{classical_independent_option_pricing}, define the matrices $Q := \left(q\super 1,\cdots, q\super K\right)$ and $V := \left(V\super 1,\cdots, V\super K\right)$, and assume quantum matrix access $\QA({\rm vec}(Q) )$ and $\QA({\rm vec}(V) )$.
\end{problem}

\subsection {Problem statements for CVA} \label{cva_problem}
We may view CVA as an adjustment of the marked-to-market value of a derivative portfolio to account for counterparty credit risk, which can be calculated as the difference between the risk-free portfolio value proposed by the BSM model and its value taking into account the possibility of default. We shall take them as a fraction of the expected positive exposure to our counterparty at the time of default. In particular, the fraction must be the $LGD$ value, for we shall be compensated by this expected loss. A detailed derivation of the CVA problem and formula is outlined in \cref{cva_derivation}.

Using the terminologies as introduced in \cref{preliminaries}, the CVA computation may be expressed as
\begin{align}
\mathbb E\left[\CVA(t_0)\right] = (1 - R) \sum_{i=0}^{N-1} \phi \left(t_i < \tau < t_{i+1}\right) \cdot \mathbbm E\left[\max(V_{t_i}, 0) | \mathcal{F}_{t_0}\right],\label{cva_formula}
\end{align}
where $t_N = T$ and $T$ is set to be at least as great as the latest expiry date of the $K$ derivative securities in our portfolio. Here $V_{t_i}$, the exposure (or value) at $t_i$ of the future portfolio value has already been discounted to $t_0$. The probability of default on the contract after the contract itself has matured is set to zero.

We give formal definitions of the problem for CVA in the classical and quantum settings.

\begin{problem}[Classical CVA with finite event space] \label{problemCVA}

We are given a derivative portfolio $V$ with the longest maturity $T$, where $T$ is a positive integer.  We introduce a time discretization,  such that we have $0 = t_0 < \cdots < t_N = T$ and $t_i$ represents the time period $[t_{i-1},  t_{i})$.
We assume that there exists a known integer $n$ such that for each $t \in [T]$, we have the pair $(\Omega\super t, \Sigma\super t)$, where  $\Omega\super t = \{ x : x\in \{0,1\}^n \}$ describes the set of economic events and $\Sigma^{(t)}$ is the sigma algebra for $\Omega\super t$. For each $t \in[T]$, we define a joint distribution measure $\mathbbm Q\super t$ for the default time $\tau$ and the economic event $\omega$. Assume independence, such that their distribution measure factors into $\mathbbm Q_\tau\super t$ and $\mathbbm Q_\omega \super t$ respectively. The joint probability measures shall be described via the vectors $q \super t  \in [0,1]^N$ for which $N=2^n$.
For each time $t$, we have a discounted exposure to the counterparty $V\super t: \Omega\super t \to (-\infty,\infty)$ and $(V\super   t)^{+} = \max\left(0, V\super t\right)$. There exists a known credit default recovery rate $R >0$. The CVA for discounted exposure $\left(V\super t\right)^{+}$ under economic event $j$ corresponding to interval $t$ is given
\begin{align}
\CVA_j \super t := \mathbbm{1}[(\tau, \omega) \in (t, j)] (1 - R) \left(V\super t_j \right)^{+}.
\end{align}
The expected CVA for the portfolio corresponding to time t shall be given by
\begin{align}
\mathbbm E_{q\super t}[\CVA\super t] = (1 - R) \sum_{j=1}^N q\super t_j \left(V\super t_j\right)^{+}.
\end{align}
Define the random variable of the total credit valuation adjustment $\CVA := \sum_{t=1}^T \CVA\super t$. The task is to evaluate
\begin{align}
\label{cva_equivalency}
\mathbb E_{q \super 1 \cdots q \super T}\left[\CVA\right] =\sum^T_{t = 1} \mathbbm E_{q\super t}\left[\CVA\super t\right].
\end{align}

For the inputs to the discounted portfolio values, we are given element-wise access to elements $V_j \super t$ with a single query of cost one for all $j \in [N]$. For the input to the joint probability measures, we consider two scenarios:
\begin{enumerate}
    \item 
     We are given element-wise access to elements $q_j \super t$ with a single query of cost one. 
    \item
    We are given element-wise sampling access to elements $q_j \super t$ with a single query of cost one. 
\end{enumerate}

\end{problem}
Note that the \cref{cva_equivalency} is equivalent to the formulation in \cref{cva_formula}, where we have used the assumption that their joint distribution measure factors to define CVA under the expectation with respect to probability measure $\mathbbm Q\super t$.

\begin{problem}[Quantum CVA with finite event space] \label{problemCVAQuantum}
Let $c$ be a chosen, positive integer such that $\forall j, t,  V\super t_j \in \mathbb{R}$ may be represented up to desired accuracy using fixed point encoding as in \cref{encoding}. Given the setting in \cref{problemCVA}, define the matrices $Q := (q\super 1 ,\cdots, q\super T)$ and $V := (V\super 1,\cdots, V\super T)$,
assume 
\begin{enumerate}
\item quantum matrix access $\QA({\rm vec} (Q))$ and $\QA({\rm vec} (V))$,
\item
quantum multi sampling access $\QS({\rm vec} (Q))$ and $\QA({\rm vec} (V))$, and knowledge of $\Vert \rm vec(Q) \Vert_1$.
\end{enumerate}
We further assume all oracle access costs of $\mathcal T_{\QA}$ and $\mathcal T_{\QS}$ are $1$.
\end{problem}

\section{Quantum Subroutines} \label{subroutines}

\subsection{General Quantum Subroutines}
We provide quantum subroutines useful for tackling the problem statements formulated in the earlier section. The following is a generalised lemma on outputs of arbitrary randomized (classical and quantum) algorithms that provide us a lower bound on success probabilities.

\begin{lemma}[Powering Lemma \cite{jerrum1986}] \label{powerLemma}
Let $A$ be a randomized algorithm estimating some quantity $\mu$.  Let the output of one pass of $A$ be denoted $\hat{\mu}$ satisfying $|\mu - \hat{\mu}| \leq \epsilon$ with probability $(1 - \gamma)$, for some $\gamma < 0.5$.  Then for any $\delta \in (0, 1)$,  repeating $A$ $\Ord {\log \left (\frac{1}{\delta}\right)}$ times and taking the median gives an estimate accurate to error $\epsilon$ with probability $\geq 1 - \delta$.
\end{lemma}

Our first quantum algorithm is related to finding the minimum or maximum entry in an arbitrary vector. Note that we may implement quantum maximum finding by an equivalent algorithm with trivial modifications. We will see that this is often useful; when we are given an arbitrary algorithm over a set of numbers with the preconditions requiring a maximum value of one, we may fulfill such conditions on arbitrary sets of numbers by finding its maximum and dividing each element of the set by this value.

\begin{lemma}[Quantum minimum finding \cite{durr1996}] \label{lemmaMin}
Let there be given quantum access to a vector $\vb{u} \in \mathbb{R}^N$ via $\QA(\vb{u})$ on $\Ord{c + \log N}$ qubits, where $c$ is the number of bits used in the fixed point binary encoding.
We may obtain by using quantum search techniques $i_{\min} = \argmin_{j\in[N]}  u_j$ with probability $> \frac{1}{2}$ using $\Ord{\sqrt{N}}$ queries and $\tOrd{c \sqrt{N}}$ additional quantum gates. By \cref{powerLemma},  we can find the minimum $i_{\min} = \argmin_{j\in[N]}  u_j$ with success probability $1-\delta$ with $\Ord{\sqrt N \log \left (\frac{1}{\delta}\right) }$ queries and $\tOrd{c \sqrt N  \log \left( \frac{1}{\delta}\right )}$ quantum gates.  Accessing the minimum value $\vb{u}[i_{\min}] = u_{\min}$ costs $\Ord 1$.
\end{lemma}

Next, we give the result on unbiased amplitude estimation due to \cite{cornelissen_sublinear-time_2023} and prove a corollary from it to suit our use cases.


\begin{theorem}[Unbiased amplitude estimation~\cite{cornelissen_sublinear-time_2023}]
\label{thm:unbiased_amp_est}
Let $\ket{\psi} = U \ket{\bar{0}}$ be a quantum state prepared by $U$ and let $\Pi$ be a projector with $a = \| \Pi \ket{\psi} \|^2$.
Given $K \geq 4$ and $\epsilon \in (0, 1)$, the unbiased amplitude estimation outputs an estimate $\tilde{a} \in [-2\pi, 2\pi]$ such that 
\begin{align}
    |\EE[\tilde{a}] - a| &\leq \epsilon & \Var[\tilde{a}] \leq \frac{91a}{K^2} + \epsilon \ .
\end{align}
using $\order{K\log\log(K) \log(K/\epsilon)}$ applications in expectation of the reflection operators $\id - 2\dyad{\psi}$ and $\id - 2\Pi$.
One copy of the state $\ket{\psi}$ is used, which is restored at the end of the computation with probability at least $1 - \epsilon$.
\end{theorem}

To make it more convenient to use in our case, we convert it into the following corollary using Chebyshev's inequality.

\begin{cor}
\label{cor:unbiased_amp_est}
Let $U \ket{\bar{0}} = \ket{\psi}$ be a quantum state prepared by $U$ and let $\Pi$ be a projector with $a = \| \Pi \ket{\psi} \|^2$.
Given $\epsilon, \delta \in (0, 1)$, the unbiased amplitude estimation outputs an estimate $\tilde{a} \in [-2\pi, 2\pi]$ such that 
\begin{align}
    |\EE[\tilde{a}] - a| &\leq \Ord{\epsilon^2} & \Pr(|\tilde{a} - a| \leq \epsilon) \geq 1 - \delta \ ,
\end{align}
using $\Ord{\frac{\sqrt{a}}{\epsilon} \log{\frac{1}{\delta}} \log\log{\frac{\sqrt{a}}{\epsilon}} \log{\frac{\sqrt{a}}{\epsilon^3}}}$ applications in expectation of the reflection operators $\id - 2\dyad{\psi}$ and $\id - 2\Pi$.
\end{cor}

\begin{proof}
    Using the triangle inequality, we have
    \begin{align}
        |\tilde{a} - a| \leq |\tilde{a} - \EE[\tilde{a}]| + |\EE[\tilde{a}] - a| \ .
    \end{align}
    From \cref{thm:unbiased_amp_est}, one can achieve $|\EE[\tilde{a}] - a| \leq 9a/K^2$ and $\Var[\tilde{a}] \leq 100 a/K^2$ using $\order{K \log \log(K) \log(K^3/9a)}$ applications of the reflection operators $\id - 2\dyad{\psi}$ and $\id - 2 \Pi$.
    From Chebyshev's inequality, we have
    \begin{align}
        \Pr(|\tilde{a} - \EE[\tilde{a}]| \leq \epsilon) \geq 1 - \frac{\Var[\tilde{a}]}{\epsilon^2} \geq 1 - \frac{100 a}{\epsilon^2 K^2} \ .
    \end{align}
    Set $K = \frac{15 \sqrt{a}}{\epsilon}$, which gives $|\EE[\tilde{a}] - a| \leq 9 \epsilon^2/225$ and $\Pr(|\tilde{a} - \EE[\tilde{a}]| \leq \epsilon) > 1/2$.
    Using the powering lemma, we have that applying the unbiased amplitude estimation $\Ord{\log{\frac{1}{\delta}}}$ achieves $\Pr(|\tilde{a} - \EE[\tilde{a}]| \leq \epsilon) \geq 1 - \delta$.
    The total number of applications of the reflection operators is 
    \begin{align}
        \Ord{\frac{\sqrt{a}}{\epsilon} \log{\frac{1}{\delta}} \log\log{\frac{\sqrt{a}}{\epsilon}} \log{\frac{\sqrt{a}}{\epsilon^3}}} \ .
    \end{align}

\end{proof}


\subsection{Quantum Subroutines for estimation of norms and inner products}

We use the notation $\mathbb R^+$ to denote non-negative reals $[0, \infty)$. A summary of the quantum subroutines introduced and proved is presented in \cref{tableQuantumGeneral}.

\begin{table}[H]
\begin{center}
\begin{tabular} { |c |c | c | c | c | c |c|c|c|}
\hline
Out & $\epsilon$ & Constraints & $\vb{u}$ & Access & $\vb{v}$ & Access & Queries & Lemma \\
\hline \hline
$\Vert \vb{u} \Vert_1$ & {\rm rel.} & $\max_j u_j = 1$ & $[0,1]^N$ 
& $\QA$ & n/a & n/a & $\tOrd{\frac{1}{\epsilon} \sqrt{\frac{ N}{\Vert \vb{u} \Vert_1}}  \log{\frac{1}{\delta}}}$ & ~\cref{lemmaU} \\
\hline
$\vb{u \cdot v}$ & {\rm rel.} & - & $\mathbb{R}^{+N}$ 
& $\QA$ & $\mathbb{R}^{+N}$ & $\QA$ & $\tOrd{ \frac{1}{\epsilon} \sqrt{\frac{N \left( \max_j u_j v_j \right)}{\vb{u} \cdot \vb{v}}} \log{\frac{1}{\delta}}  }$ & ~\cref{lemma_inner_product_RA} \\
\hline
$\vb{u \cdot v}$ & {\rm rel.} & - & $\mathbb{R}^{+N}$ 
& $\QS$ & $[0, 1]^N$ & $\QA$ & $\tOrd{\frac{1}{\epsilon \sqrt{\vb{u} \cdot \vb{v}}} \log{\frac{1}{\delta}}}$ & ~\cref{lemmaU_Vunit} \\
\hline
$\vb{u \cdot v}$ & {\rm add.} & - & $\mathbb{R}^{+N}$ 
& $\QS$ & $[0, 1]^N$ & $\QA$ & $\tOrd{\frac{1}{\epsilon} \log{\frac{1}{\delta}}}$ & ~\cref{lemmaU_Vunit} \\
\hline
$\vb{u \cdot v}$ & {\rm add.} & $\ell_2$-norm & $\Delta^N$ 
& $\QS$ & $\mathbb{R}^{+N}$ & $\QA$ & $\tOrd{\frac{1}{\epsilon} \log{\frac{1}{\delta}}}$ & ~\cref{bounded_norm} \\
\hline
$\vb{u \cdot v}$ & {\rm add.} & $\Var$ & $\Delta^N$ 
& $\QS$ & $\mathbb R^{+N}$ & $\QA$ & $\tOrd{\frac{\sigma}{\epsilon} \log{\frac{1}{\delta}}}$ & ~\cref{lemmaBoundedVarianceAdditive} \\
\hline
$\vb{u \cdot v}$ & {\rm rel.} & $\Var$ & $\Delta^N$ 
& $\QS$ & $\RR^{+N}$ & $\QA$ & $\tOrd{\frac{B}{\epsilon} \log{\frac{1}{\delta}}}$ & ~\cref{lemmaBoundedVarianceRelative} \\
\hline
\end{tabular}
\caption{Summary of quantum algorithms}
\label{tableQuantumGeneral}
\end{center}
\end{table}

The next lemma is for the norm estimation of a vector with non-negative entries. We assume that the vector is non-zero and has been normalized such that the largest element is one. Such a vector might be obtained by dividing first the maximum value, as discussed earlier.

\begin{lemma}[Quantum state preparation and norm estimation with relative accuracy] \label{lemmaU}
Let there be given a non-zero vector $\vb{u} \in [0, 1]^N$ and $\max_j u_j = 1$. We are given quantum access to $\vb{u}$ via $\QA(\vb{u})$. Then:

\begin{enumerate}[(i)]
    \item There exists a unitary operator that prepares the state
    $$\ket{\psi} = \frac{1}{\sqrt{N}}  \sum_{j=1}^N \ket j  \left( \sqrt{u_j } \ket{0} + \sqrt{1- u_j} \ket{1} \right)$$ with two queries and number of gates $\Ord{\log N}$.
    \item Let $\epsilon, \delta \in (0,1)$. There exists a quantum algorithm that  provides an estimate $\Gamma_{\vb{u}}$ of the $\ell_1$-norm
    $\Vert \vb{u} \Vert_1$  such that
    $\left \vert \EE[\Gamma_{\vb{u}}] - \Vert \vb{u} \Vert_1 \right \vert \leq \Ord{\epsilon^2 \Vert \vb{u} \Vert_1/N}$,
    and $\Pr( |\Gamma_{\vb{u}} - \| \vb{u} \|_1| \leq \epsilon \| \vb{u} \|_1) \geq 1 - \delta$ with query complexity $\tOrd{\frac{1}{\epsilon} \sqrt{\frac{N}{\| \vb{u} \|_1}} \log{\frac{1}{\delta}}}$.
\end{enumerate}
\end{lemma}
\begin{proof}
For (i),  prepare a uniform superposition of all $\ket j$ with $\Ord{\log N}$ Hadamard gates.
Suppose the fixed point encoding of $\vb{u}$ consists of $c$ bits.
With the quantum query access, perform
\begin{align*}
\frac{1}{\sqrt{N}} \sum_{j=1}^N \ket j \ket {0^{c + 1}} &\xrightarrow{\QA(\vb{u})} \frac{1}{\sqrt{N}}  \sum_{j=1}^N \ket j  \ket{ u_j} \ket { 0} \\
&\xrightarrow{\text{Controlled rotation}} \frac{1}{\sqrt{N}}  \sum_{j=1}^N \ket j  \ket{ u_j}  \left( \sqrt{u_j} \ket{0} + \sqrt{1-u_j} \ket{1} \right) \\
&\xrightarrow{\QA(\vb{u})} \frac{1}{\sqrt{N}}  \sum_{j=1}^N \ket j  \ket{0^c}  \left( \sqrt{u_j} \ket{0} + \sqrt{1-u_j} \ket{1} \right)
\end{align*}
The steps consist of two oracle queries and a controlled rotation. The controlled rotation (see \cref{defQCR}) is well-defined as $ 0 \leq u_j \leq 1$ and costs $\Ord{ 1}$ gates. 
Then, discarding the $\ket{0^c}$ gives the desired state $\ket{\psi}$.

For (ii),
define a unitary $\mathcal U = \id - 2 \dyad{\psi}$, which can be constructed with the unitary from (i). Define another unitary by $\mathcal V = \mathbbm 1-2 \mathbbm 1 \otimes \ket{0} \bra{0}$.
Consider the quantity
$a := \mel{\psi}{(\id \otimes \dyad{0})}{\psi} = \frac{\Vert \vb{u} \Vert_1}{N }$, for which $1/N \leq a\leq 1$, since $u_{\max} =1$ by assumption. 
The unbiased amplitude estimation algorithm in \cref{cor:unbiased_amp_est} gives an estimate $\tilde{a}$ such that $\Pr(|\tilde{a} - a| \leq \epsilon a) \geq 1 - \delta$ and $|\EE [\tilde{a}] - a| \leq \Ord{\epsilon^2 a^2}$, using $\Ord{\frac{1}{\epsilon \sqrt{a}} \log{\frac{1}{\delta}} \log\log{\frac{1}{\epsilon \sqrt{a}}} \log{\frac{1}{\epsilon^3 a^{5/2}}}}$ applications of $\mathcal U$ and $\mathcal V$ in expectation.
Let $\Gamma_{\vb{u}} := N \tilde{a}$, which gives $|\EE[\Gamma_{\vb{u}}] - \| \vb{u} \|_1| \leq \Ord{\epsilon^2 \| \vb{u} \|_1/N}$ and query complexity
\begin{align}
    \Ord{\frac{1}{\epsilon} \sqrt{\frac{N}{\| \vb{u} \|_1}} \log{\frac{1}{\delta}} \log\log(\frac{1}{\epsilon } \sqrt{\frac{N}{\| \vb{u} \|_1}}) \log(\frac{1}{\epsilon^3} \left( \frac{N}{\| \vb{u} \|_1} \right)^{5/2})} = \tOrd{\frac{1}{\epsilon} \sqrt{\frac{N}{\| \vb{u} \|_1}} \log{\frac{1}{\delta}}} \ .
\end{align}


\end{proof}

The following lemma allows us to estimate the inner products between two vectors of arbitrary, non-negative entries. 
If one of the vectors is discretized probability density mass and the other contains the values of a corresponding random variable, we see that this is leads to computing expectation values.

\begin{lemma} [Quantum inner product estimation with relative accuracy] \label{lemma_inner_product_RA}
Let $\epsilon,\delta \in(0,1)$.  Let there be two vectors $\vb{u}, \vb{v} \in [0, \infty)^N$. We are given quantum access to $u_j, v_j$ via $\QA(\vb{u}), \QA(\vb{v})$ respectively.
Then, knowing the value of  $z_{\max} := \max_j u_j v_j$, an estimate $I$ for the inner product can be provided such that $\Pr(|I - \vb{u} \cdot \vb{v}| \leq \epsilon \vb{u} \cdot \vb{v}) \geq 1 - \delta$ and $|\EE[I] - \vb{u} \cdot \vb{v}| \leq \Ord{\frac{\epsilon^2 \vb{u} \cdot \vb{v}}{N}}$.
This output is obtained with $\tOrd{\frac{1}{\epsilon} \sqrt{\frac{N z_{\max}}{\vb{u} \cdot \vb{v}}} \log{\frac{1}{\delta}}}$ query complexity.
\end{lemma}

\begin{proof}
Define the vector $\vb{z}$ such that $z_j = u_j v_j$ via quantum oracles $\QA(\vb{u})$ and $\QA(\vb{v})$.  Then, we have $\Vert{\vb{z}} \Vert_1 = \vb{u} \cdot \vb{v}$.

If $z_{\max} = 0$, the estimate for the inner product is $I = 0$ and we are done.
Otherwise, apply \cref{lemmaU} with the vector $ \frac{ \vb{z}}{ z_{\max}}$ to obtain an estimate $\tilde{a}$ of the norm $a := \left \Vert  \frac{ \vb{z}}{ z_{\max}} \right \Vert_1$ such that
$\left \vert \EE[\tilde{a}] - a \right \vert \leq \Ord{\frac{\epsilon^2 a}{N}}$,
and $\Pr( |\tilde{a} - a| \leq \epsilon a) \geq 1 - \delta$ with query complexity $\tOrd{\frac{1}{\epsilon} \sqrt{\frac{N}{a}} \log{\frac{1}{\delta}}} = \tOrd{\frac{1}{\epsilon} \sqrt{\frac{N z_{\max}}{\vb{u} \cdot \vb{v}}} \log{\frac{1}{\delta}}}$.
Set $I := z_{\max} \tilde{a}$, which gives $\Pr(|I - \vb{u} \cdot \vb{v}| \leq \epsilon \vb{u} \cdot \vb{v}) \leq 1 - \delta$ and $|\EE[I] - \vb{u} \cdot \vb{v}| \leq \Ord{\frac{\epsilon^2 \vb{u} \cdot \vb{v}}{N}}$.


\end{proof}

An immediate theorem results from the conclusions drawn in \cref{lemma_inner_product_RA} to estimate the sum of element-wise products of two matrices of equivalent size.

\begin{theorem} \label{corTrace}
Given quantum access to an element-wise $c$-bit representation of the matrices $A, B \in \mathbbm [0,\infty)^{N \times M}$ according to \cref{defQM} and knowledge of $z_{\max} := \max_{j\in[N], k\in[M]} A_{jk} B_{jk}$. Then, an estimate $I$ for $t = {\rm tr} \{ A^T B \}$ can be provided such that $\vert \EE[I] - t \vert \leq \Ord{\frac{t \epsilon^2}{NM}}$ and $\Pr(|I - t| \leq \epsilon t) \geq 1 - \delta$.
This output is obtained with $\tOrd{\frac{1}{\epsilon} \sqrt{\frac{MN z_{\max}}{t}} \log{\frac{1}{\delta}}}$ query complexity.
\end{theorem}
\begin{proof}
Note that
${\rm tr} \{ A^T B \} = {\rm vec}(A) \cdot {\rm vec}(B)$.
The result follows immediately from \cref{lemma_inner_product_RA}.
\end{proof}

We provide a generalization of \cref{lemma_inner_product_RA} to allow for approximating the mean of an arbitrary vector with respect to non-uniform distributions under a sampling model. Note that the entries of the vector $\vb{v}$ are bounded between zero and one. 

\begin{lemma}[Quantum inner product estimation with sampling access] \label{lemmaU_Vunit}
Let $N$ be a positive integer and let $\vb{u} \in \Delta^N$, $\vb{v} \in [0, 1]^N$ be two non-zero vectors. 
We are given quantum access to $u_j, v_j$ via $\QS(\vb{u})$ and  $\QA(\vb{v})$, respectively.
Then:
\begin{enumerate}[(i)]
\item There exists a unitary operator that prepares the state
$$ \ket{\psi} = \sum_{j=1}^N \sqrt{u_j} \ket j  \left( \sqrt{v_j} \ket{0} + \sqrt{1- v_j} \ket{1} \right)$$ with three queries and number of gates $\Ord{\log N}$.

\item Let $\epsilon, \delta \in (0,1)$. There exists a quantum algorithm that  provides an estimate $\Gamma$ such that
$\Pr(\left| \Gamma - \vb{u} \cdot \vb{v} \right| \leq \epsilon \vb{u} \cdot \vb{v}) \geq 1 - \delta$ and $\left \vert \vb{u} \cdot \vb{v} - \EE[\Gamma] \right \vert \leq \Ord{\epsilon^2 ({\vb{u} \cdot \vb{v}})^2}$.
The query complexity of the algorithm is $\tOrd{\frac{1}{\epsilon} \sqrt{\frac{1}{\vb{u} \cdot \vb{v}}} \log{\frac{1}{\delta}}}$.

\item Let $\epsilon, \delta \in (0,1)$. There exists a quantum algorithm that provides an estimate $\Gamma$ such that
$\left \vert \vb{u} \cdot \vb{v} - \EE[\Gamma] \right \vert \leq \Ord{\epsilon^2}$ and $\Pr(\left| \Gamma - \vb{u} \cdot \vb{v} \right| \leq \epsilon) \geq 1 - \delta$. 
The algorithm requires $\tOrd{\frac{1}{\epsilon} \log{\frac{1}{\delta}}}$ queries.
\end{enumerate}
\end{lemma}

\begin{proof}
For (i), with the quantum query access and the quantum sampling access, perform
\begin{align*}
    \ket {\bar 0} &\xrightarrow{\QS(\vb{u}), \QA(\vb{v})} \sum_{j=1}^N \sqrt{u_j} \ket j  \ket{ v_j} \ket { 0}  \\
    &\xrightarrow{\text{Controlled rotation}} \sum_{j=1}^N \sqrt{u_j}\ket j  \ket{ v_j}  \left( \sqrt{v_j} \ket{0} + \sqrt{1 - v_j} \ket{1} \right) \\
    &\xrightarrow{\QA(\vb{v})} \sum_{j=1}^N \sqrt{u_j}\ket j  \ket{ \bar{0}}  \left( \sqrt{v_j} \ket{0} + \sqrt{1-v_j} \ket{1} \right)
\end{align*}
The rotation is well-defined as $0 \leq v_j \leq 1$ and costs $\Ord{\log{N}}$ gates. 
Then discarding the register $\ket{\bar{0}}$ gives the desired state $\ket{\psi}$.

For (ii), define a unitary $\mathcal U = \mathbbm 1 - 2 \dyad{\psi}$, which can be constructed from the unitary operations from (i). 
Define another unitary by $\mathcal V = \mathbbm 1-2 \mathbbm 1 \otimes \dyad{0}$.
Let $a := \mel{\psi}{(\id \otimes \dyad{0})}{\psi} = \vb{u} \cdot \vb{v}$.
Invoking \cref{cor:unbiased_amp_est}, there is an unbiased amplitude estimation algorithm to output an estimate $\Gamma$ such that $\left \vert \vb{u} \cdot \vb{v} - \EE[\Gamma] \right \vert \leq \Ord{\epsilon^2 ({\vb{u} \cdot \vb{v}})^2}$ and $\Pr(\left| \Gamma - \vb{u} \cdot \vb{v} \right| \leq \epsilon \vb{u} \cdot \vb{v}) \geq 1 - \delta$ using $\tOrd{\frac{1}{\epsilon} \sqrt{\frac{1}{\vb{u} \cdot \vb{v}}} \log{\frac{1}{\delta}}}$ applications of $\mathcal{U}$ and $\mathcal{V}$.



For (iii), invoking \cref{cor:unbiased_amp_est}, there is an unbiased amplitude estimation algorithm to output an estimate $\Gamma$ such that $\left \vert \vb{u} \cdot \vb{v} - \EE[\Gamma] \right \vert \leq \Ord{\epsilon^2}$ and $\Pr(\left| \Gamma - \vb{u} \cdot \vb{v} \right| \leq \epsilon) \geq 1 - \delta$ using $\tOrd{\frac{1}{\epsilon} \sqrt{\vb{u} \cdot \vb{v}} \log{\frac{1}{\delta}}} = \tOrd{\frac{1}{\epsilon} \log{\frac{1}{\delta}}}$ applications of $\mathcal U$ and $\mathcal V$, where we have used $\vb{u} \cdot \vb{v} \leq 1$.


\end{proof}

We provide lemmas on estimation of inner products on vectors with arbitrary entries subject to bounded $\ell_2$-norm. We relax the assumption on entries of the vector $\vb{v}$ such that it is only bounded from below by zero and has a finite representation. This lemma is a vectorized form of the equivalent result on random variables by Montanaro \cite{Montanaro2015}, with reworking by applying the unbiased amplitude estimation from \cite{cornelissen_sublinear-time_2023} instead of the original amplitude estimation~\cite{brassard2002quantum}.

\begin{lemma}[Quantum inner product estimation on vectors of bounded $\ell_2$-norm with additive accuracy] \label{bounded_norm}
Assume that we are given non-zero vectors $\vb{u} \in \Delta^N$ and $\vb{v} \in [0, +\infty)^N$.  
Define a vector $\vb{w}$ such that $w_j = v_j^2 u_j$. Suppose that we are guaranteed that $\Vert \vb{w} \Vert_1$ is upper bounded by some constant of $\Ord{1}$. We are given quantum access to $u_j, v_j$ via $\QS(\vb{u}), \QA(\vb{v})$ respectively. 
Then:

Let $\epsilon, \delta \in (0,1)$. There exists a quantum algorithm that provides an estimate $\Gamma$ of
$\vb{v} \cdot \vb{u}$  such that
$\Pr( \left \vert \vb{v} \cdot \vb{u} - \Gamma\right \vert \leq \epsilon (\sqrt{\Vert \vb{w} \Vert_1} + 1)^2) \geq 1 - \delta$ and $\left| \EE[\Gamma] - \vb{u} \cdot \vb{v} \right| \leq \Ord{\epsilon^2 \log{\frac{1}{\epsilon}} + \epsilon \| \vb{w} \|_1}$. The algorithm requires 
$\tOrd{\frac{1}{\epsilon} \log{\frac{1}{\delta}}}$ queries.

\end{lemma}

\begin{proof}

Let $k = \ceil{\log_2{1/\epsilon}}$ with $\epsilon < \frac{1}{2}$. 
Then,  for $l \in [k]$,  let a unitary $U_{l}$ be defined by the following transformation,
\begin{align*}
\ket{\bar 0} \xrightarrow{\QS(\vb{u})} \sum_{j=1}^N  \sqrt{u_j} \ket j \ket 0 \to  \sum_{j=1}^N \sqrt{u_j} \ket j  \ket{v_{l,j}}  =: \ket{\chi_{l}} \ ,
\end{align*}
where $v_{l, j}$ is defined as
\be
v_{l,j} :=
\begin{dcases}
    v_j, & \text{if } l = 0,  0 \leq v_j < 1 \\
    \frac{v_j}{2^l}, & \text{if } 0 < l \leq k, 2^{l - 1} \leq v_j < 2^l\\
    0, & \text{otherwise}.
\end{dcases}
\ee

The steps defining the unitary consist of oracle queries $\QS(\vb{u}), \QA(\vb{v})$, $\QC(\vb{v}, l)$ and a division of the second register by $2^l$, which may be implemented efficiently via \cref{rls}.

Let $\vb{v}_{l} = (v_{l, 1}, \cdots, v_{l, N})$ for $0 \leq l \leq k$ and let $\vb{v}_{> k} := (v_{>k, 1}, \cdots, v_{>k, N})$, where $v_{>k, j} = v_j$ if $v_j \geq 2^{k}$.
Then, $\vb{v} = \sum_{l = 0}^k 2^{l} \vb{v}_l + \vb{v}_{> k}$ and
\begin{align}
    \vb{u} \cdot \vb{v} = \sum_{l = 0}^k 2^{l} \vb{u} \cdot \vb{v}_l + \vb{u} \cdot \vb{v}_{> k} \ .
\end{align}
Since $0 \leq v_{l, j} \leq 1$ for al $j$ and $l$, we can use \cref{lemmaU_Vunit} to estimate $\vb{u} \cdot \vb{v}_l$.
Specifically, for each $0 \leq l \leq k$, \cref{lemmaU_Vunit} gives us an estimate $\tilde{a}_l$ such that 
\begin{align}
    \Pr(\left| \tilde{a}_l - \vb{u} \cdot \vb{v}_l \right| \leq \epsilon_l \sqrt{\vb{u} \cdot \vb{v}_l}) &\geq 1 - \delta_l & \left| \EE[\tilde{a}_l] - \vb{u} \cdot \vb{v}_l \right| \leq \Ord{\epsilon_l^2 \vb{u} \cdot \vb{v}_l} \ ,
\end{align}
with query complexity $\Ord{\frac{1}{\epsilon_l} \log{\frac{1}{\delta_l}}}$.
Here, we set $\delta_0 = \frac{1}{10}$ and $\delta_l = \frac{1}{10k}$ for $1 \leq l \leq k$; $\epsilon_0 = \epsilon$ and $\epsilon_l = \frac{\epsilon}{\sqrt{k}}$ for $1 \leq l \leq k$.

Now, define an estimator for $\vb{u} \cdot \vb{v}$ to be $\Gamma = \sum_{l = 0}^k 2^l \tilde{a}_l$.
Then, 
\begin{align}
    \left| \Gamma - \vb{u} \cdot \vb{v} \right| &\leq \sum_{l = 0}^k 2^l \left| \tilde{a}_l - \vb{u} \cdot \vb{v}_l \right| + \vb{u} \cdot \vb{v}_{> k} \\
    &\leq \epsilon \sqrt{\vb{u} \cdot \vb{v}_0} + \sum_{l=1}^k 2^l \frac{\epsilon}{\sqrt{k}} \sqrt{\vb{u} \cdot \vb{v}_l} + \vb{u} \cdot \vb{v}_{> k} \ .
\end{align}
For the second term,
\begin{align}
    \vb{u} \cdot \vb{v}_{> k} = \sum_{j=1}^N u_j v_{>k, j} \leq \frac{1}{2^k} \sum_{j = 1}^N u_j v_{>k, j}^2 \leq \frac{\| \vb{w} \|_1}{2^k} \ ,
\end{align}
where $w_j := u_j v_j^2$.
For the first term, 
\begin{align}
    \frac{\epsilon}{\sqrt{k}} \sum_{l = 1}^k 2^l \sqrt{\vb{u} \cdot \vb{v}_l} & \leq \epsilon \left( \sum_{l=1}^k 2^{2l} \vb{u} \cdot \vb{v}_l \right)^{1/2} \leq \epsilon \sqrt{2 \|\vb{w}\|_1} \ .
\end{align}
Above, the last inequality is because $2^{l} v_{l, j} \geq 2^{l-1}$, which implies that $2^{2l-1} u_j v_{l, j} \leq 2^{2l} u_j v_{l, j}^2$. 
Moreover, $\vb{u} \cdot \vb{v}_0 \leq 1$.
Thus, 
\begin{align}
    \left| \Gamma - \vb{u} \cdot \vb{v} \right| &\leq \epsilon \left( 1 + \sqrt{2} \sqrt{\| \vb{w} \|_1} + \| \vb{w} \|_1 \right) \leq \epsilon (\sqrt{\| \vb{w} \|_1} + 1)^2 \ .
\end{align}

On the other hand, 
\begin{align}
    \left| \EE[\Gamma] - \vb{u} \cdot \vb{v} \right| &\leq \sum_{l = 0}^k 2^l \left| \EE[\tilde{a}_l] - \vb{u} \cdot \vb{v}_l \right| + \vb{u} \cdot \vb{v}_{> k} \\
    &\leq \epsilon^2 \vb{u} \cdot \vb{v}_0 + \sum_{l=1}^k \frac{\epsilon^2}{k} \vb{u} \cdot \vb{v}_l + \frac{\| \vb{w} \|_1}{2^k} \\
    &\leq \Ord{\epsilon^2 \log{\frac{1}{\epsilon}} + \epsilon \| \vb{w} \|_1} \ .
\end{align}

Finally, we calculate the complexity bound and the success probability.
The total query complexity is 
\begin{align}
    \Ord{\frac{1}{\epsilon} + k \frac{\sqrt{k}}{\epsilon} \log{k}} = \tOrd{\frac{1}{\epsilon}} \ .
\end{align}
By union bound, the success probability is at least $4/5$.
Applying the above procedure $\log{\frac{1}{\delta}}$ times and taking the median increase the success probability to $1 - \delta$, and the total query complexity becomes $\tOrd{\frac{1}{\epsilon} \log{\frac{1}{\delta}}}$. 
\end{proof}

We provide lemmas to relax the constraint on estimation of inner products on vectors with arbitrary entries subject to bounded variance $\sigma^2$. The following two lemmas provide the similar results by Montanaro \cite{Montanaro2015} under the assumption of vector inputs.
The difference is that we use the unbiased amplitude estimation from \cite{cornelissen_sublinear-time_2023} as the workhorse, instead of the conventional one~\cite{brassard2002quantum}.

\begin{lemma}[Quantum inner product estimation on inputs of bounded variance with additive accuracy] \label{lemmaBoundedVarianceAdditive}
Given a positive constant $N$, assume that we are given two non-zero vectors $\vb{u} \in \Delta^N$ and $\vb{v} \in \RR^{+N}$.  
Define vector $\vb{w}$ such that $w_j = v_j^2 u_j$. Suppose that for some known quantity $\sigma >0$, we are guaranteed that $\Vert \vb{w} \Vert_1 \leq \sigma^2$. We are given quantum access to $u_j, v_j$ via $\QS(\vb{u}), \QA(\vb{v})$ respectively. 
Then:

Let $\epsilon,\delta \in (0,1)$. There exists a quantum algorithm that  provides an estimate $\Gamma$ of
$\vb{v} \cdot \vb{u}$  such that
$\Pr( \left \vert \vb{v} \cdot \vb{u} - \Gamma\right \vert \leq \epsilon) \geq 1 - \delta$ and $\left| \EE[\Gamma] - \vb{u} \cdot \vb{v} \right| \leq \Ord{\frac{\epsilon^2}{\sigma} \log{\frac{\epsilon}{\sigma}} + \epsilon}$. The algorithm requires $\tOrd{\frac{\sigma}{\epsilon} \log{\frac{1}{\delta}}}$ queries.
\end{lemma}

Our proof is similar that of Montanaro (Theorem 2.5. \cite{Montanaro2015}), with extra analysis on the bias.

\begin{proof}
 
Let $U_{\chi}$ be the unitary operator such that
\begin{align*}
\ket{\bar 0} \xrightarrow{\QS(\vb{u})} \sum_{j=1}^n  \sqrt{u_j} \ket j \ket{0}  \xrightarrow{\QA(\vb{v})}  \sum_{j=1}^N \sqrt{u_j} \ket j  \ket{v_j}  =: \ket{\chi}.
\end{align*}
Let $m$ be the random variable resulting from the measurement of the second register on $\ket \chi$, such that for all $j$, $\Pr( m = v_j) = u_j$. 
We have $\Var(m) = \mathbb{E}[(m- \mathbb{E}[m])^2] \leq \mathbb{E}[m^2] = \Vert \vb{w} \Vert_1 \leq \sigma^2 $.
Moreover,  let $U_{\chi'}$ be the  unitary operator such that
\begin{align}
\ket{\bar 0} \xrightarrow{\QS(\vb{u})} \sum_{j=1}^n  \sqrt{u_j} \ket j \ket 0 \to \sum_{j=1}^N \sqrt{u_j} \ket j  \ket{v_j'}  =: \ket{\chi'},\nonumber
\end{align}
where $v_j' := \frac{v_j}{\sigma}$.  
The second step consists of one application of $\QA(\vb{v})$ followed by dividing the second register by $\sigma$, which may be implemented efficiently via \cref{rls}.
Correspondingly,  let $m'$ be the random variable resulting from the measurement of the second register on $\ket{\chi'}$, such that for all $j$, $\Pr( m' = \frac{v_j}{\sigma}) = u_j$.  
It follows that $\Var(m') = \frac{\Var(m)}{\sigma^2}\leq 1$.
Let $m_0$ be the result of a random sampling from the random variable $m'$.   
Applying Chebyshev's inequality, we have $\Pr(\vert m' - \mathbb{E}[m'] \geq 3) \leq \frac{1}{9}$. 

Similarly,  let $U_{\tilde{\chi}'}$ be the unitary operator such that
\begin{align}
\ket{\bar 0} \xrightarrow{\QS(\vb{u})} \sum_{j=1}^n  \sqrt{u_j} \ket j \ket 0 \to \sum_{j=1}^N \sqrt{u_j} \ket j  \ket{\tilde{v_j}'}  =: \ket{\tilde{\chi}'},\nonumber
\end{align}
where $\tilde{v_j}' := \frac{v_j - \sigma m_0}{4 \sigma}$.  
The second step consists of one application of $\QA(\vb{v})$ followed by subtracting the second register by $\sigma m_0$ and dividing by $4\sigma$.  
These operations can be implemented efficiently.  
Let $\tilde{m}'$ be the random variable resulting from the measurement of the second register on $\ket{\tilde{\chi}'}$, such that for all $j$, $\Pr( \tilde{m}' = \frac{v_j - \sigma m_0}{4 \sigma}) = u_j$.  
Then, with probability at least $8/9$,
\begin{align*}
    \EE [\tilde{m}'^2]^{1/2} &= \left( \sum_j \left( \frac{v_j}{4\sigma} - \frac{m_0}{4} \right)^2 u_j \right)^{1/2} \\
    &\leq \frac{1}{4} \left( \sum_j \left( \frac{v_j}{\sigma} - \EE [m'] \right)^2 u_j \right)^{1/2} + \frac{1}{4} \left( \sum_j (\EE [m'] - m_0)^2 u_j \right)^{1/2} \\
    &\leq \frac{1}{4\sigma} \left( \sum_j v_j^2 u_j \right)^{1/2} + \frac{3}{4} \\
    &\leq \frac{1}{4\sigma} \sqrt{\| \vb{w} \|_1} + \frac{3}{4} \leq 1 \ .
\end{align*}

Note that $\tilde{v}_j'$ is not necessarily positive.
To apply \cref{bounded_norm}, we need to construct positive random variables from it.
Define the unitary operator $U_{-}$ that maps $\ket{\tilde{v}'_j} \to \ket{-\tilde{v}'_j}$.  
Then,  let $U_{\tilde{\chi}_+'}$ be the unitary operator such that 
\begin{align}
\ket{\bar 0} \xrightarrow{\QS(\vb{u})} \sum_{j=1}^n  \sqrt{u_j} \ket j \ket 0 \xrightarrow{\QA( \tilde{\vb{v}}'_+)}  \sum_{j=1}^N \sqrt{u_j} \ket j  \ket{\tilde{v}'_{j, +}}  =: \ket{\tilde{\chi}_+'},\nonumber
\end{align}
and $U_{\tilde{\chi}_-'}$ be the unitary operator such that
\begin{align*}
\ket{\bar 0} \xrightarrow{\QS(\vb{u})} \sum_{j=1}^n  \sqrt{u_j} \ket j \ket{0} \xrightarrow{\QA( \tilde{\vb{v}}{'} _-)}  \sum_{j=1}^N \sqrt{u_j} \ket j  \ket{\tilde{v}'_{j, -}}  =: \ket{\tilde{\chi}_-'},
\end{align*}
where $\tilde{v}'_{j,+} = \tilde{v}_j'$  if  $\tilde{v}_j' \geq 0$ and $0$ otherwise.  
Similarly, $\tilde{v}'_{j,-} = -\tilde{v}'_j$  if  $\tilde{v}'_j \leq 0$ and 0 otherwise.
Here, $U_{\tilde{\chi}_+'}$ may be implemented by one application of $U_{\tilde{\chi}'}$ and invoking \cref{quantum decomposition}. Also, $U_{\tilde{\chi}_-'}$ may be implemented by one application of $U_{\tilde{\chi}'}$,  $U_{-}$ on the second register and invoking \cref{quantum decomposition}.

Let $\tilde{m}'_+, \tilde{m}'_-,$ be the random variable resulting from the measurement of the second register on $\ket{\tilde{\chi}_+'},  \ket{\tilde{\chi}_-'}$ respectively.  
Up to this point, we recall that $\tilde{\vb{v}} = \frac{1}{4\sigma} \vb{v} - \frac{m_0}{4} \vb{1}$, which means that
\begin{align}
    \vb{v} = \sigma (m_0 \vb{1} + 4 \tilde{\vb{v}}') = \sigma (m_0 \vb{1} + 4 \tilde{\vb{v}}'_+ - 4 \tilde{\vb{v}}'_-) \ .
\end{align}
Moreover, define two vectors $\tilde{\vb{w}}'_+$ and $\tilde{\vb{w}}'_-$ by $\tilde{w}'_{j,+} := \tilde{v}'^{2}_{j,+} u_j$ and $\tilde{w}'_{j,-} := \tilde{v}'^{2}_{j,-} u_j$.
Then, it is not hard to show that $\| \tilde{\vb{w}}'_+ \|_1 \leq \EE [\tilde{m}'^2] \leq 1$ and $\| \tilde{\vb{w}}'_- \|_1 \leq \EE [\tilde{m}'^2] \leq 1$.
Since $\mathbb E\left[\tilde{m}'_-\right] = \tilde{\vb{v}}'_- \cdot \vb{u}$ and $\mathbb E\left[\tilde{m}'_+\right] = \tilde{\vb{v}}'_+ \cdot \vb{u}$, we may invoke \cref{bounded_norm} to output estimate $\tilde{\mu}_+$ and $\tilde{\mu}_-$ of $\mathbb E\left[\tilde{m}'_+\right]$ and $\mathbb E\left[\tilde{m}'_-\right]$, respectively, with accuracy $\frac{\epsilon}{32\sigma}$ and bias $\Ord{\frac{\epsilon^2}{\sigma^2} \log{\frac{\epsilon}{\sigma}} + \frac{\epsilon}{\sigma}}$, using $\tOrd{\frac{\sigma}{\epsilon} \log{\frac{1}{\delta}}}$ queries.
From these, we can construct estimator $\Gamma := \sigma (m_0 + 4 \tilde{\mu}_+ - 4 \tilde{\mu}_-)$ of $\vb{u} \cdot \vb{v}$.
It is not hard to compute that the accuracy is $\left| \Gamma - \vb{u} \cdot \vb{v} \right| \leq \epsilon$ with probability $>1/2$, which can be boosted to $1 - \delta$ using the powering lemma (\cref{powerLemma}).
The bias is given by $\left| \EE[\Gamma] - \vb{u} \cdot \vb{v} \right| \leq \Ord{\frac{\epsilon^2}{\sigma} \log{\frac{\epsilon}{\sigma}} + \epsilon}$.
The query complexity is $\tOrd{\frac{\sigma}{\epsilon} \log{\frac{1}{\delta}}}$.

\end{proof}

\begin{lemma}[Quantum inner product estimation on inputs of bounded variance with relative accuracy] 
\label{lemmaBoundedVarianceRelative}

Assume we are given non-zero vectors $\vb{u} \in \Delta^N$ and $\vb{v} \in \RR^{+N}$.  
Define a vector $\vb{w}$ such that $w_j = v_j^2 u_j$, and a vector $\vb{z}$ such that $z_j = v_j u_j$. 
Suppose that for some known quantity $B$, we are guaranteed that $\frac{\Vert \vb{w} \Vert_1}{\Vert \vb{z} \Vert_1^{2}} \leq  B $. We are given quantum access to $u_j, v_j$ via $\QS(\vb{u}), \QA(\vb{v})$ respectively. 
Then:

Let $\epsilon,\delta \in (0,1)$. There exists a quantum algorithm that  provides an estimate $\Gamma$ of
$\vb{v} \cdot \vb{u}$ such that
$\Pr( \left \vert \vb{v} \cdot \vb{u} - \Gamma\right \vert \leq \epsilon \vb{v} \cdot \vb{u}) \geq 1-\delta$ and $\left \vert \vb{v} \cdot \vb{u} - \EE[\Gamma] \right \vert \leq \epsilon' \vb{v} \cdot \vb{u}$, where $\epsilon' = \Ord{\frac{\epsilon^2}{B^2} \log{\frac{B}{\epsilon}} + \epsilon B}$. The algorithm requires $\tOrd{\frac{B}{\epsilon} \log{\frac{1}{\delta}}}$ queries.
\end{lemma}

Yet again, our proof follows similar line to that of Montanaro (Theorem 2.6. \cite{Montanaro2015}), with the extra analysis on bias.

\begin{proof}
Let $U_{\chi}$ be the unitary operator such that 
\begin{align*}
\ket{\bar 0} \xrightarrow{\QS(\vb{u})} \sum_{j=1}^n  \sqrt{u_j} \ket j \ket 0 \xrightarrow{\QA(\vb{v})}  \sum_{j=1}^N \sqrt{u_j} \ket j  \ket{ v_j}  =: \ket{\chi}.
\end{align*}
Let $m$ be the random variable resulting from the measurement of the second register on $\ket \chi$, such that for all $j$, $\Pr(m = v_j) = u_j$. 
Note that $\EE[m] = \vb{u} \cdot \vb{v}$ and $\Var(m) = \mathbb{E}[(m - \mathbb{E}[m])^2] \leq \mathbb{E}[m^2] = \Vert \vb{w} \Vert_1$.
 
Let $k := \ceil{32B}$ and let $\tilde m = \frac{1}{k} \sum^k_{i=1} m_i$ be the mean of $k$ samples of $m$ obtained via independent measurements of the last $c_2$ second register on $\ket \chi$.  Then, the expectation $\mathbb{E}[\tilde{m}] = \mathbb{E}[m] = \Vert \vb{z} \Vert_1$ and the variance is $\Var(\tilde{m}) = \frac{1}{k} \Var(m)$.
Since we are guaranteed that $\frac{\Var(m)}{\mathbb{E}[m]^2} \leq \frac{\Vert \vb{w} \Vert_1}{\Vert \vb{z} \Vert_1^2} \leq B$, it follows that $\Var(\tilde{m}) \leq \frac{\mathbb{E}[m]^2 B}{k} \leq \frac{\mathbb{E}[m]^2}{32}$. 
Applying the Chebyshev's inequality, $\Pr\left(\vert \tilde m - \mathbb E[\tilde m] \vert \geq \frac{|\EE [\tilde{m}]|}{2} \right) \leq \frac{4 \Var(\tilde{m})}{\EE[\tilde{m}]^2} \leq \frac{1}{8}$. 
Thus, 
\begin{align}
    \frac{\EE[m]}{2} \leq \tilde{m} \leq \frac{3}{2} \EE[m] 
\end{align}
with probability at least $7/8$.

Moreover,  let $U_{\chi'}$ be the unitary operator such that 
\begin{align*}
\ket{\bar 0} \xrightarrow{\QS(\vb{u})} \sum_{j=1}^n  \sqrt{u_j} \ket j \ket 0 \to \sum_{j=1}^N \sqrt{u_j} \ket j  \ket{v_j'}  =: \ket{\chi'},
\end{align*}
where $v_j' := \frac{v_j}{\tilde{m}}$. The steps to implement this unitary consist of one application of $U_{\chi}$ followed by dividing the second register by $\tilde{m}$, which may be implemented efficiently via \cref{rls}. 
Let $m'$ be the random variable resulting from the measurement of the second register on $\ket {\chi'}$, such that for all $j$, $\Pr( m' = \frac{v_j}{\tilde{m}}) = u_j$. 
We also have $\EE[m'^2] = \frac{\EE[m^2]}{\tilde{m}^2} = \frac{\| \vb{w} \|_1}{\tilde{m}^2}$.
Invoking \cref{bounded_norm}, we have an estimate $\Gamma'$ of $\EE[m']$ with additive error
\begin{align}
    \left| \Gamma' - \EE[m'] \right| &\leq \frac{2\epsilon}{3 (2\sqrt{B} + 1)^2} (\sqrt{\EE[m'^2]} + 1)^2 \\
    &\leq \frac{\EE[m] \epsilon}{\tilde{m} (2\sqrt{B} + 1)^2} \left( \frac{\sqrt{\| \vb{w} \|_1}}{\tilde{m}} + 1 \right)^2 \\
    &\leq \frac{\EE[m] \epsilon}{\tilde{m} (2\sqrt{B} + 1)^2} \left( \frac{2 \sqrt{\| \vb{w} \|_1}}{\EE[m]} + 1 \right)^2 \\
    &= \frac{\EE[m] \epsilon}{\tilde{m} (2\sqrt{B} + 1)^2} \left( \frac{2 \sqrt{\| \vb{w} \|_1}}{\| \vb{z} \|_1} + 1 \right)^2 \leq \frac{\EE[m] \epsilon}{\tilde{m} } \ ,
\end{align}
and success probability at least $7/8$,
where we have used $\frac{\Vert \vb{w} \Vert_1}{\Vert \vb{z} \Vert_1^2} \leq B$.
The query complexity is $\tOrd{\frac{B}{\epsilon}}$.
The bias is given by
\begin{align}
    \left| \EE[\Gamma'] - \EE[m'] \right| &\leq \Ord{\frac{\epsilon^2}{B^2} \log{\frac{B}{\epsilon}} + \epsilon \frac{\| \vb{w} \|_1}{\tilde{m}^2}} \\
    &\leq \Ord{\frac{\epsilon^2}{B^2} \log{\frac{B}{\epsilon}} + \epsilon \frac{\| \vb{w} \|_1}{\EE[m]^2}} \\
    &\leq \Ord{\frac{\epsilon^2}{B^2} \log{\frac{B}{\epsilon}} + \epsilon B} \ .
\end{align}
By union bound, the total failure probability is upper bounded by $1/4$.

Define $\Gamma := \tilde{m} \Gamma'$. We have $\left| \Gamma - \EE[m] \right| \leq \epsilon \EE[m]$ with probability at least $3/4$ and $\left| \EE[\Gamma] - \EE[m] \right| \leq \epsilon' \EE[m]$, where $\epsilon' = \Ord{\frac{\epsilon^2}{B^2} \log{\frac{B}{\epsilon}} + \epsilon B}$.
Using the powering lemma (\cref{powerLemma}), we can boost the success probability to $1 - \delta$, and the query complexity becomes $\tOrd{\frac{B}{\epsilon} \log{\frac{1}{\delta}}}$.

\end{proof}

\section{Solutions to the Problem Statements} \label{solutions}
We now use the quantum subroutines proven earlier to tackle the problem statements.

\subsection{Quantum algorithm for the multi-asset portfolio pricing and the CVA problem}
The lemmas discussed can be used to solve these problems under the query access model. Recall that for quantum multi-option pricing in \cref{problemCountableQuantum}, we are given quantum matrix access to $Q := \left(q\super 1,\cdots, q\super K\right)$ and $V := (V\super 1,\cdots, V\super K)$ via oracles $\QA({\rm vec}(Q) )$ and $\QA({\rm vec}(V) )$.

\begin{theorem}[Quantum multi-asset portfolio pricing]\label{thm:multi-option-pricing}
Consider \cref{problemCountableQuantum}.
Then, the value of $z_{\max} := \max\limits_{j \in [N], k \in [K]} q_{j}\super k V_{j}\super k$ and an estimate $I$ for ${\mathbb{E}[\TV] }$ can be provided such that $|\EE[I] - \EE[\TV]| \leq \Ord{\frac{\EE[\TV] \epsilon^2}{NK}}$ and that $\vert I - \mathbb E[\TV] \vert \leq \epsilon\ \mathbb E[\TV]$ with success probability $1-\delta$.
This output is obtained with $\tOrd{\left [ \sqrt {NK} + \frac{1}{\epsilon} \sqrt{\frac{NK z_{\max}}{\mathbb E[\TV]}} \right]   \log \left (\frac{1}{\delta} \right )  }$ queries.
\end{theorem}
\begin{proof}
First, we can use the input to find $z_{\max} := \max_{jk} q_{j}\super k V_{j}\super k$ using \cref{lemmaMin}, with success probability $1 - \frac{\delta}{2}$. This takes $\Ord{\sqrt{NK} \log(\frac{1}{\delta})}$ queries and $\tOrd{\sqrt{NK} \log(\frac{1}{\delta})}$ gates.

Note that $\mathbb E[\TV] = \sum_{k=1}^K \sum_{j=1}^N q\super k_j V\super k_j = {\rm vec}(Q) \cdot {\rm vec}(V)$.
Employing \cref{corTrace} with the quantum matrix access to $Q$ and $V$, we obtain an estimate $I$ such that $|\EE[I] - \EE[\TV]| \leq \Ord{\frac{\EE[\TV] \epsilon^2}{NK}}$ and $\Pr(|I - \EE[\TV]| \leq \epsilon \EE[\TV]) \geq 1 - \frac{\delta}{2}$ with query complexity $\tOrd{\frac{1}{\epsilon} \sqrt{\frac{KN z_{\max}}{\EE[\TV]}} \log{\frac{1}{\delta}}}$.
Via the union bound, the result follows.
\end{proof}

A similar technique can be applied to the CVA problem.
Consider the CVA Problem Statement in \cref{cva_problem}. We have the formulation
\begin{align}
\mathbb E_{q \super 1 \cdots q \super T}[\CVA] =(1 - R)  \sum^T_{t = 1} \sum_{j=1}^N q\super t_j \left(V\super t_j\right)^{+}.
\end{align}

We consider the discussion of the CVA problem under settings of no additional information about its moments. Recall that for the quantum CVA setting in \cref{problemCVAQuantum}, we are given quantum matrix access to $Q := (q\super 1,\cdots, q\super K)$ and $V := (V\super 1,\cdots, V\super K)$ via oracles $\QA({\rm vec}(Q) )$ and $\QA({\rm vec}(V) )$.

\begin{theorem}[Quantum single-asset credit valuation adjustment]
Consider \cref{problemCVAQuantum}, Setting 1. Then, the value of $z_{\max} := \max\limits_{j \in [N], t \in [T]} q_{j}\super t \left(V_{j}\super t\right)^+$ and an estimate $I$ for $\mathbb E[\CVA]$ can be provided such that $|\EE[I] - \EE[\CVA]| \leq \Ord{\frac{\epsilon^2 \EE[\CVA]}{NT}}$ and $\vert I - \mathbb E[\CVA] \vert \leq \epsilon\  \mathbb E[\CVA]$ with success probability $1-\delta$.
This output is obtained with $\Ord{\left [\sqrt {NT} + \frac{1}{\epsilon} \sqrt{\frac{(1 - R)NT z_{\max}}{\mathbb E[\CVA]}} \right]   \log \left (\frac{1}{\delta} \right )  }$ queries.
\end{theorem}

\begin{proof}
With the quantum access $\QA({\rm vec}(V))$ we can invoke \cref{quantum decomposition} to obtain quantum access to the non-negative part of the vector, i.e., $\QA({\rm vec}\left(V^+\right))$. First, we can use the input to find $z_{\max} := \max_{jk} q_{j}\super t \left(V_{j}\super t\right)^+$ using \cref{lemmaMin}. This takes $\Ord{\sqrt{NT} \log(\frac{1}{\delta})}$ queries and $\tOrd{\sqrt{NT} \log(\frac{1}{\delta})}$ gates.

Note that i) $0 \leq R \leq 1$ and ii) ${\rm \frac{\mathbb E[\CVA]}{1 - R} } = \sum_{t=1}^T \sum_{j=1}^N q\super t_j \left(V\super t_j\right)^+ = {\rm vec}(Q) \cdot {\rm vec}(V^+)$. Employing \cref{corTrace} with the quantum matrix access to $Q$ and $(V)^+$, we obtain an estimate $I$ such that $|\EE[I] - \EE[\CVA]| \leq \Ord{\frac{\epsilon^2 \EE[\CVA]}{NT}}$ and $\Pr(|I - \EE[\CVA]| \leq \epsilon \EE[\CVA]) \geq 1 - \frac{\delta}{2}$ with query complexity $\tOrd{\frac{1}{\epsilon} \sqrt{\frac{(1 - R)TN z_{\max}}{\EE[\CVA]}} \log{\frac{1}{\delta}}}$.
Via the union bound, the result follows.

\end{proof}

\subsection{Quantum algorithm for the CVA problem in the Black-Scholes-Merton setting}
We may consider the discussion of the CVA problem under settings of additional constraints up to the second order moments. 

In most cases, some information is known about the distribution of future asset prices. We introduce the theory of asset pricing relevant to the CVA pricing. Financial derivatives have payoff functions that are dependent on the trajectory of the underlying assets. There is significant literature \cite{hull2015options} behind modelling these asset price dynamics, the most influential of which is the work of Black, Scholes and Merton (BSM) \cite{Black1973,Merton1973}. The BSM model derives the price of an option on an asset modelled as a geometric Brownian motion. In particular, the dynamics of a stock price $S_t$ is captured by the SDE \cite{Shreve2004}
\begin{align}
\dd{S_t} = \alpha S_t \dd t + \sigma S_t \dd W_t,
\end{align}
where $\dd W_t$ is the Brownian increment, $\alpha$ is the drift, and $\sigma$ is the volatility of the asset price. The Brownian motion $W_t$ is defined under some probability measure $\mathbb P$. Using Ito's Lemma, it can be shown that the SDE can be solved:
\begin{align}
S_t = S_{0}\exp\left\{\sigma W_t + \left(\alpha - \frac{\sigma ^2}{2}\right)t\right\}.
\end{align}
We introduced in the \cref{subsec:financial_preliminaries} the concept of discounting, which is used to determine the present value of future asset valuations. Consider the money market/bank account $B_t$. Assume that the interest rates are constant. For some interest rate $r \in (0, 1)$, investing in the money market account has value $B_t = B_0 \exp\{rt\}$ at time $t$. Assume $B_0 = 1$. The First Fundamental Theorem of Asset Pricing states that the principle of no-arbitrage is equivalent to the existence of a risk-neutral measure; discounted asset prices are martingales under the risk-neutral measure $\mathbb Q$. In particular, as the model economy is adapted to the filtration $\mathcal F_t$, we have
\begin{align}
\frac{S_t}{B_t} = \mathbb E_\mathbb Q{\left[\frac{S_T}{B_T} | \mathcal F_t\right]}.
\end{align}
The asset price dynamics under the risk-neutral measure can be written as
\begin{align}
\dd S_t = r S_t \dd t + \sigma S_t \dd \tilde{W_t},
\end{align}
where $\tilde W_t$ is $\mathbb Q$-Brownian.

We are interested in obtaining the present value of a derivative. Specifically, the price of a derivative at time $t$ should be $\exp\{-r(T-t)\}\mathbb E_\mathbb Q{[f(S_T) | \mathcal F_t]}$, where $f$ is the payoff function of the underlying asset $S$ maturing at time $T$. In the case of simple payoff functions, the solutions to the SDE can be determined analytically. However, in the case of multivariate portfolios or complex payoff functions as in exotic derivatives, no closed-form solutions exist. Instead, numerical approaches such as Monte Carlo methods or finite difference approximation schemes are used. Monte Carlo sampling provides a general approach and is an integral pricing tool in a derivatives desk, allowing not only for complex payoffs but also to model other stylized facts such as joint dependencies, heavy-tailed distributions \cite{tankov2003} and fractional Brownian motion \cite{zhu2021}. Furthermore, when other parameters such as the volatility or interest rates are modelled as stochastic processes, the Monte Carlo approach is favored. The abstract view of valuing a derivative portfolio can be outlined as such :

\begin{enumerate}[(i)]
 
\item Sample paths of asset price dynamics under the risk-neutral measure $\mathbb Q$ calibrated to market variables at $\mathcal F_0$.
\item Compute asset prices of each path.
\item Compute the derivative payoff using payoff functions $f$.
\item Take the mean $\tilde{\mu}$ of the discounted payoffs over the samples.
\item The portfolio value/price is approximated by this expected value. The variance $\tilde{\lambda}^2$ of the portfolio value is the sample variance of the paths' payoffs.
\end{enumerate}

We note that this is the problem tackled in \cref{problemCountableQuantum}. A more detailed and nuanced approach is expounded upon in pricing literature under quantum settings \cite{Rebentrost2018finance}. 
For completeness, we review the pricing formula that obtains exactly the expectation and variance of the portfolio when the portfolio consists of a single, European call option. While such settings are simplifications of practical concerns in derivatives practice, the analytical models provide a useful benchmark for which numerical approaches can be compared to.

The European call option is the right, but not the obligation to purchase an underlying asset $S$ at some future maturity time $T$ at some pre-determined strike price $K$. The option payoff function is given 
\be f(S, K, T) := \max(S_T - K, 0) = (S_T - K)^+.\ee
For any $T > t$, this is necessarily random and we are tasked with pricing the option at $t$, denoted $v(t,S(t)) = \exp\{-r(T-t)\}\mathbb E_\mathbb Q{[(S_T - K)^+ | \mathcal F_t]}$. Feynman-Kac theorem asserts that $v(t,S(t))$ satisfies the BSM partial differential equations. The value of the European call can be shown to have the solution \cite{Shreve2004}
\be
\mathrm{BSM}(e, x, K, r, \sigma) := x \Phi(d_+(e, x)) - K\exp\{-re\} \Phi(d_-(e, x)),
\ee
where
\be
d_{\pm}(e, x) = \frac{1}{\sigma\sqrt{e}} \left[\log{\frac{x}{K}} + e \left(r \pm \frac{\sigma^2}{2}\right)\right],
\ee
and $\Phi$ is the c.d.f. of a standard normal, $e = (T-t)$ is the time to maturity, $x$ the current price, $K$ the strike price, $r$ the discount rate and $\sigma$ the asset volatility.

The asset price of an exponentiated Brownian motion has log-normal distributions, and the variance of the European call under the risk-neutral measure $\mathbb Q$ can be computed exactly. When $t = 0$, the variance of the payoff can be shown (Lemma 4, \cite{Rebentrost2018finance}) :
\be
\Var(f(S_T)) = \exp\left\{2rT + T\sigma^2\right\}S_0^2 \Phi\left(\tilde{d}(T, x)\right) - 2K\exp\{rT\} S_0 \Phi\left(d_+(T, x)\right) \nonumber \\
+ K^2 \Phi(d_-(T, x)) - \left(S_0 \exp\{rT\} \Phi(d_+(T, x)) - K \Phi(d_-(T, x))\right)^2,
\ee
where
\be
\tilde{d}(T, x) = \frac{1}{\sigma\sqrt{T}} \left[\log{\frac{x}{K}} + T \left(r + \frac{3\sigma^2}{2}\right)\right].
\ee
Let $\Var(f(S_T))$ be upper bounded by $\tilde{\lambda}^2$, which we know to be a fairly low order polynomial in $S_0, K, \exp \{ rT \}$ and $\exp \{ \sigma^2 T \}$. In particular, $$\Var(f(S_T)) \leq \exp\{2rT + T\sigma^2\}S_0^2 + K^2 =: \tilde\lambda^2,$$ since the Gaussian probabilities are upper bounded by one.
The expected value of the CVA is given in \cref{cva_formula}. Under the BSM settings for a European call option, the random variable may be expressed:
\be
\frac{\CVA}{(1-R)} \equiv 
\sum_{t=1}^T \phi\left(\tau \in [t]\right) 
\left(V\super t\right)^+ = \sum_{t=1}^T \phi(\tau \in [t])  \exp\{-rt\} f(S_t),
\ee
where the last term is the BSM portfolio value at time $t$. Accordingly, the variance of the CVA can be bounded
\be
\frac{\Var(\CVA)}{(1-R)^2} &\leq& 
\sum_{i=1}^T \sum_{j=1}^T \phi(\tau \in [t_i]) \phi(\tau \in [t_j]) 
\cdot \mathrm{CoV}\left[f(S_{t_i}), f(S_{t_j})\right] 
\\ &\leq& T^2 \cdot \max_{k \in [T]} \Var \left(f\left(S_{t_k}\right)\right).
\ee
In the special case of the European call option, we have $\Var(\CVA) \leq T^2 (1-R)^2 \tilde{\lambda}^2$.

We further consider sampling access to $q \super t$, which we have defined to be determined via the joint probability measures $\mathbbm Q\super t$ factored by into $\mathbbm Q_\tau\super t$ and $\mathbbm Q_\omega \super t$. The probability value $q_j \super t$ can be decomposed \cite{Zapata2021} :
\be
q_j \super t := \phi(\tau \in [t]) \cdot \mathbb P_{\mathbb Q}(V_j, t) = \phi(\tau \in [t]) \cdot \mathbb P_{\log \mathcal N}(V_j | t) \cdot \mathbb P_{\mathcal U}(t),
\ee 
where $\mathbb P_{\mathcal U}(t)$ the uniform probability measure over $[T]$, and $\mathbb P_{\log \mathcal N}(V_j | t)$ is log-normal for an underlying modelled as an exponentiated Brownian motion.
Specifically, 
\be \mathbb P_{\mathcal U}(t) = \frac{1}{T} \ee
and
\be
\mathbb P_{\log \mathcal N}(V_j | t) := \frac{1}{\sigma V_j \sqrt{2 \pi}} \exp\left\{ -\frac{\left( \ln\frac{V_j}{V_0} - \left(\mu - \frac{\sigma^2}{2}\right)t \right) ^2 }{2\sigma^2 t} \right\}.
\ee
Using the formulas defined above, define
\be
\Vert \mathrm{vec(Q)} \Vert_1 := 
\sum_{t=1}^T \sum_{j=1}^N q\super t_j = \frac{1}{T} \sum_{t=1}^T \sum_{j=1}^N \phi(\tau \in [t]) \cdot \mathbb P_{\log \mathcal N}(V_j | t) \leq 1.
\ee
The default probabilities $\phi(\tau \in [t])$ are computable by bootstrapping credit curves \cite{giuseppe}. The bootstrapping technique is covered in the \cref{cds_bootsrapping}.

Recall that for the quantum CVA setting in \cref{problemCVAQuantum}, Setting 2, we assumed quantum multi-sampling and multi-vector access. Furthermore, $\Vert \rm vec(Q) \Vert_1$ is assumed to be known, which we have argued is a reasonable assumption under practical conditions. In the case of the European call option, we also have an exact bound on the CVA variance, and that $\Var(\CVA) \leq T^2 (1-R)^2 \tilde{\lambda}^2$. In general settings, we assume that a similar upper bound on $\Var(\CVA)$ is known using Monte Carlo or equivalent techniques.
 
\begin{theorem}[Quantum credit valuation adjustment on bounded variance  to additive error] \label{theorem_additive}
Consider \cref{problemCVAQuantum}, Setting 2. For some constant $\sigma > 0$ and recovery rate $R \in (0, 1)$,  we suppose that the $\CVA$ has bounded variance such that $\Var(\CVA) \leq \sigma^2 (1 - R)^2$. Then, the estimate $I$ for $\mathbb E[\CVA]$ can be provided such that $\left| \EE[I] - \EE[\CVA] \right| \leq \Ord{\frac{\epsilon^2}{(1-R) \sigma} \log{\frac{\epsilon}{(1-R)\sigma}} + \epsilon}$ and that $\vert I - \mathbb E[\CVA] \vert \leq \epsilon $ with success probability $1-\delta$.
This output is obtained with query complexity $\tOrd{\frac{\sigma(1-R)}{\epsilon} \log{\frac{1}{\delta}}}$.
\end{theorem}

\begin{proof}
Note that i) $0 \leq R \leq 1$ and ii) 
\begin{align}
{\rm \frac{\mathbb E[\CVA]}{1 - R} } = \sum_{t=1}^T \sum_{j=1}^N q\super t_j \left(V\super t_j\right)^+
= {\rm vec}(Q) \cdot {\rm vec}\left(V^+\right). 
\end{align} 
Note that
$\Var\left(\frac{\CVA}{1-R}\right) \leq \sigma^2$.
With the quantum access $\QA({\rm vec}(V))$ we can invoke \cref{quantum decomposition} to obtain quantum access to the non-negative part of the vector, i.e., $\QA({\rm vec}(V^+))$. 
Together with quantum sampling access $\QS({\rm vec}(Q))$, we can apply \cref{lemmaBoundedVarianceAdditive} to obtain an estimate $\Gamma$ of $\frac{\mathbb E[\CVA]}{1-R}$ such that $\left| \EE[\Gamma] - \frac{\EE[\CVA]}{1 - R} \right| \leq \Ord{\frac{\epsilon^2}{\sigma} \log{\frac{\epsilon}{\sigma}} + \epsilon}$ and $\Pr( \left \vert \frac{\mathbb E[\CVA]}{1- R} - \Gamma \right \vert \leq \epsilon ) \geq 1 - \delta$ with query complexity $\tOrd{\frac{\sigma}{\epsilon} \log{\frac{1}{\delta}}}$.
The result immediately follows by letting $I = (1-R)\Gamma$ and setting $\epsilon \to \frac{\epsilon}{1 - R}$. 
\end{proof}

\begin{theorem}[Quantum credit valuation adjustment on bounded variance  to relative error] \label{theorem_relative}
Consider \cref{problemCVAQuantum}, Setting 2.  For some constant $B > 0$ and recovery rate $R \in (0, 1)$,  we suppose that the $\CVA$ has bounded variance such that $\Var(\CVA) \leq B \cdot \mathbb{E}[{\CVA}]^2$. Then, the estimate $I$ for $\mathbb E[\CVA]$ can be provided such that $\left| \EE[I] - \EE[\CVA] \right| \leq \epsilon' \frac{\EE[\CVA]}{1-R}$ and that $\vert I - \mathbb E[\CVA] \vert \leq \epsilon \mathbb E[ \CVA] $ with success probability $1-\delta$, where $\epsilon' = \Ord{\frac{\epsilon^2}{B^2} \log{\frac{B}{\epsilon}} + \epsilon B}$.
This output is obtained with $\tOrd{\frac{B}{\epsilon} \log{\frac{1}{\delta}}}$ queries.
\end{theorem}

\begin{proof}
Note that i) $0 \leq R \leq 1$ and ii) 
\begin{align}
{\rm \frac{\mathbb E[\CVA]}{1 - R} } = \sum_{t=1}^T \sum_{j=1}^N q\super t_j \left(V\super t_j\right)^+
= {\rm vec}(Q) \cdot {\rm vec}\left(V^+\right). 
\end{align}
Additionally, note that
\begin{align}
\frac{ \Var\left(\frac{\CVA}{1-R}\right) } {\mathbb{E}\left[{\frac{\CVA}{1- R}}\right]^2} = \frac{\frac{\Var\left(\CVA\right)}{(1 - R)^2}}{\frac{\mathbb E[\CVA]^2}{(1 - R)^2}} = \frac{\Var(\CVA)}{\mathbb E[\CVA]^2} \leq B.
\end{align}
With the quantum access $\QA({\rm vec}(V))$ we can invoke \cref{quantum decomposition} to obtain quantum access to the non-negative part of the vector, i.e., $\QA({\rm vec}(V^+))$.
Together with quantum sampling access $\QS({\rm vec}(Q))$, we can apply \cref{lemmaBoundedVarianceRelative} to obtain an estimate $\Gamma$ of $\frac{\mathbb E[\CVA]}{1-R}$ such that $\left| \EE[\Gamma] - \frac{\EE[\CVA]}{1-R} \right| \leq \epsilon' \frac{\EE[\CVA]}{1-R}$ and $\Pr( \left \vert \frac{\mathbb E[\CVA]}{1- R} - \Gamma \right \vert \leq \frac{\epsilon \mathbb E[\CVA]}{1- R} ) \geq 1 - \delta$ with query complexity $\tOrd{\frac{B}{\epsilon} \log{\frac{1}{\delta}}}$, where $\epsilon' = \Ord{\frac{\epsilon^2}{B^2} \log{\frac{B}{\epsilon}} + \epsilon B}$.
The result immediately follows by multiplying the result $\Gamma$ by $(1 - R)$.

\end{proof}

\section{Discussion and Conclusion} \label{conclusions}

\subsection{Potential quantum advantage}

Here, we discuss the potential quantum advantage compared to classical Monte Carlo methods for solving multi-option pricing and CVA.
For multi-option pricing (\cref{classical_independent_option_pricing}), suppose we are given access to $V = (V^{(1)}, \cdots, V^{(K)})$ and $Q = (q^{(1)}, \cdots, q^{(K)})$ via the vector access $\VA(\lvec(V))$ and $\VA(\lvec(Q))$.
Then, one can use the $\ell_1$-sampling technique from \cite{Tang2018} to derive the query complexity for classical Monte Carlo estimation.
Specifically, define a random variable $Z := V_i \| \lvec(Q) \|_1$ with probability $\frac{Q_i}{\| \lvec(Q) \|_1}$, where $i \in [NK]$.
Note that $\| \lvec(Q) \|_1 = K$, since each $q^{j}$ is a normalized probability vector for $j \in [K]$.
Then, $\EE[Z] = \sum_i V_i Q_i = \EE[\TV]$ and $\Var(Z) \leq \sum_i V_i^2 Q_i \| \lvec(Q) \|_1 \leq V_{\max} \| \lvec(Q) \|_1 \EE[\TV]$, where $V_{\max} := \max_{j \in [N], k \in [K]} V_j^{k}$.
Then, as in \cite[Proposition 4.2]{Tang2018}, using the standard median of means technique, one can derive the query complexity bound for relative error $\epsilon$ to be $\Ord{\frac{V_{\max} K}{\epsilon^2 \EE[\TV]} \log{\frac{1}{\delta}}}$, where $1 - \delta$ is the success probability. The classical lower bound for estimating the parameter of a Bernoulli random variable to multiplicative error is $\Omega\left(\frac{1}{\epsilon^2 p}\right)$ from Chebyshev's inequality. 
Comparing this lower bound to \cref{thm:multi-option-pricing}, the quantum algorithm achieves a quadratic speedup asymptotically in terms of $\epsilon$ and $\EE[\TV]$ for multi-option pricing, and the above classical algorithm suggests also an asymptotic quadratic speedup for $K$.

For the CVA problem (\cref{problemCVA}, setting 2), first suppose the variance bound is given by $\Var(\CVA) \leq \sigma^2 (1 - R)^2$ as in \cref{theorem_additive}. 
Then, with similar analysis as above, the query complexity for classical Monte Carlo estimation of $\EE[\CVA]$ with additive error $\epsilon$ is given by $\Ord{\frac{\sigma^2 (1 - R)^2}{\epsilon^2} \log{\frac{1}{\delta}}}$.
Next, for the relative error estimation, suppose the $\CVA$ has bounded variance $\Var(\CVA) \leq B \cdot \EE[\CVA]^2$ as in \cref{theorem_relative}.
Then, to achieve $\epsilon$ relative error, classical Monte Carlo requires $\Ord{\frac{B^2}{\epsilon^2} \log{\frac{1}{\delta}}}$ queries.
Compared with \cref{theorem_additive} and \cref{theorem_relative}, the discussed quantum algorithms also suggest a quadratic speedup over the classical counterparts.

\subsection{Contributions}
In summary, we have shown demonstrable improvements over current literature by presenting quantum algorithms for multi-option pricing and obtaining unbiased approximation of the CVA problem under settings of bounded variance. 
We argue that the assumptions of knowledge about the probability distributions with respect to default and portfolio processes are reasonable and obtainable under financial settings. 
By using the Quantum Minimum Finding subroutine and the unbiased Amplitude Estimation under the access model, we find that QMC accelerates multi-option pricing and the approximation of the CVA.
Compared to classical Monte Carlo estimation, our quantum algorithms give rise to quadratic speedups over a set of parameters of interest.

\subsection{Future Work}
We believe there are multiple directions and provide recommendations towards future work in quantum settings for CVA and in related topics. Under the CVA setting, heuristics and techniques to reduce circuit depth \cite{Zapata2021} may be employed, and its performance analysed for its application on near-term quantum devices. Additionally, financial literature on CVA is more extensive \cite{Green2015}, and we may extend quantum literature to account for bilateral credit risks, for example.

We may consider the quantum speedup of other components in the XVA, such as the Margin Valuation Adjustments (MVA). The calculation of MVA involves the use of regression techniques such as the Longstaff-Schwartz least-squares Monte Carlo method (LSMC), which was recently quantized \cite{doriguello2021quantum}. This speedup could translate to a speedup for MVA, which can be explored in future work.

\section*{Acknowledgement}
JYH thanks Dr. Rahul Jain and Dr. Steven Halim for their helpful comments in the review of this work. 
PR, BC, and DLV acknowledge support by the National Research Foundation, Singapore, and A*STAR under its CQT Bridging Grant and its Quantum Engineering Programme under grant NRF2021-QEP2-02-P05.

\textbf{Conflict of interest.} The authors declare no conflicts of interest. The authors alone are responsible for the content and writing of the paper.

\bibliographystyle{rQUF}
\bibliography{qxva}

\appendices

\section{More on CVA} \label{append_CVA}
XVA is a term used to encompass a series of value adjustments to the valuation of a portfolio. These adjustments are dependent on the profiles of the parties in question, usually the seller of the derivatives contracts and a corresponding buyer.

\begin{figure}[h]         
  \centering \includegraphics[width=15cm,height=15cm,keepaspectratio]{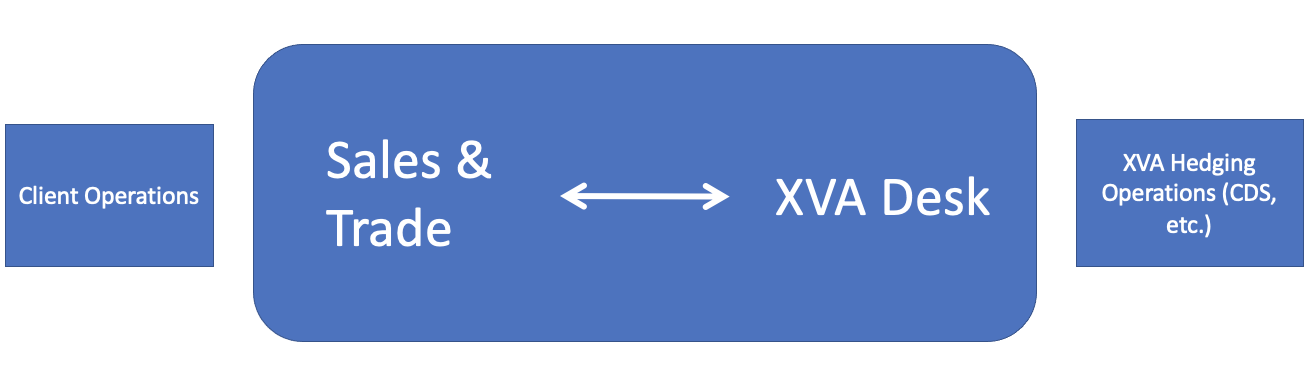}
    \caption{A XVA Trading Desk (adapted from \cite{Zeitsch2017})} \label{fig:xva_trading}  
\end{figure}

The role of the XVA desk in a trading operation is outlined in the \cref{fig:xva_trading}. The XVA desks play a primary role of valuing these adjustments, and writes protection against losses of the derivative trading operations. 
They might also optionally warehouse or hedge against these risks. 
In our discussion, we narrowed the scope of the XVA problem to just credit risks, which are risks that a party binded by a financial contract fails to make due payments to the other party. Credit Valuation Adjustment (CVA), can then be defined as the market price of credit risk on a portfolio of instruments that are marked to market. In particular, 
\[
\CVA = V(\mathrm{Default Risk Free}) - V(\mathrm{Default Risky}) . 
\]
This valuation adjustment is similar to how a chronic smoker might have to pay higher premiums on insurance. 
Since the credit crises of 2007, the Basel Framework has recommended a bilateral model.
However, many banks still use unilateral CVA and the use of bilateral models remains contentious~\cite{Green2015}. In this paper, we only consider the unilateral credit adjustments, which do not take into account one's own default risk and hence are simpler to calculate. To that effect, the CVA formula in our discussion is strictly positive.

While there is no market standard, there are two types of CVAs - unilateral CVA and bilateral CVA~\cite{brigo2013}. The difference is that unilateral CVA only takes into consideration the counterparty credit risk while bilateral CVA also takes into account the credit risk of `self', or the accounting party. Many derivative contracts such as interest rate swaps involve cash flow payments in both direction based on market conditions and as such, both parties carry the risk.

\subsection{Derivation of the CVA Formula} \label{cva_derivation}

The mathematics presented are based on previous work, mainly \cite{Green2015}, with intermediate steps provided. 
Rearranging the CVA equation, we have $V(\mathrm{DefaultRisky}) = V(\mathrm{DefaultRiskFree}) - \mathrm{CVA}$. As before, we define the terminologies:

\begin{itemize}
  \item $V_{t}$: BSM value of a basket of derivatives at time t,
  \item ${\hat{V}}_t$: economic value of a basket of derivatives at time $t$ (adjusted for the possibility of default),
  \item $V_{t}^{+}$: positive components of a basket of derivatives at time t,
  \item $V_{t}^{-}$: negative components of a basket of derivatives at time t 
  \item $C(t, t')$: cash flows for a portfolio of trades from time $t$ to $t'$, where $C(t, t') = 0$ if $t'\leq t$.
\end{itemize}

Recall that the portfolio process is adapted to filtration $\mathcal F_t$, which comes from two sources of randomness, the stock price and the default time. Note also that the unadjusted value of the portfolio $V_t$ is the expected value of its discounted future cash flows, i.e., $V_{t} = \mathbb E[C(t, T) | \mathcal{F}_{t}]$.
Let $\tau$ be the random variable of the default time. 
At the time of default, the amount $ R(V_{\tau})^{+} + (V_{\tau})^{-}$ is recovered (under the assumption that the recovered amount is computed from the BSM value of the outstanding cash flows). 
The future cash flows from time $t$ with default time $\tau$ are given by
\begin{align}
C_{t, \tau} := C(t,\tau) \mathbbm{1}_{t \leq \tau < T} + C(t,T) \mathbbm{1}_{\tau \geq T} + \left ( R(V_{\tau})^{+} + (V_{\tau})^{-}\right) \mathbbm{1}_{t\leq \tau  < T}.
\end{align}
The contributions are according to whether the default happens before or after the maturity time $T$ and the recovery amount.
Define the expected value of the cash flows $\hat{V}_{t, \tau} :=\mathbb E\left[ C_{t,\tau} | \mathcal F_t\right]$. Note that $\hat{V}_{t, \tau}$ is a functional with respect to $\tau$, akin to $\mu_X = \mathbb E[X]$ for a random variable $X$ and its expectation value.
Note that $V_{t} = (V_{t})^{+} + (V_{t})^{-}$, and thus,  
\begin{align}
R(V_{\tau})^{+} + (V_{\tau})^{-} = -(1 - R) (V_{\tau})^{+}  + V_{\tau}.
\end{align}
Using the definition $V_{\tau} = \mathbb E[C(\tau, T) | \mathcal{F}_{\tau}]$ for the second term gives
\begin{align}
\mathbb E[V_{\tau} \mathbbm{1}_{t\leq \tau < T} | \mathcal{F}_{t}] &= \mathbb E[ \mathbb E[ (C(\tau, T) {\mathbbm{1}_{t\leq \tau < T}} | \mathcal{F}_{\tau} ] | \mathcal{F}_{t}] =  \mathbb E[C(\tau, T) {\mathbbm{1}_{t\leq \tau < T}} | \mathcal{F}_{t}],
\end{align}
where we used in the second equation that $\mathbb E[ \mathbb E[ (C(\tau, T) {\mathbbm{1}_{t\leq \tau < T}} | \mathcal{F}_{\tau}] | \mathcal{F}_{t}]=\mathbb E[C(\tau, T) {\mathbbm{1}_{t \leq \tau < T}} | \mathcal{F}_{t}]$, because $\mathcal F_t \subseteq \mathcal F_\tau$ and due to the tower property of the expectation value.
Hence, the adjusted value of the portfolio at time $t$ assuming default at $\tau$ is expressed~\cite{Green2015}:
\begin{align*}
    \hat{V}_{t, \tau} 
    &= \mathbb E[ C(t, \tau) {\mathbbm{1}_{t\leq \tau < T}} | \mathcal{F}_{t}] +
    \mathbb E[C(t, T) {\mathbbm{1}_{\tau \geq T}} | \mathcal{F}_{t}] -
    \mathbb E[(1 - R) (V_{\tau})^{+} {\mathbbm{1}_{t\leq \tau < T}} | \mathcal{F}_{t}] +
    \mathbb E[C(\tau, T) {\mathbbm{1}_{t\leq \tau < T}} | \mathcal{F}_{t}],
\end{align*}
Rearranging the terms gives
\begin{align*}
    \hat{V}_{t, \tau} &= \mathbb E[(C(t, \tau) + C(\tau, T)) {\mathbbm{1}_{t\leq \tau < T}} | \mathcal{F}_{t}] +
    \mathbb E[C(t, T) {\mathbbm{1}_{\tau \geq T}} | \mathcal{F}_{t}] - \mathbb E[(1 - R) (V_{\tau})^{+} {\mathbbm{1}_{t\leq \tau < T}} | \mathcal{F}_{t}] \\
    &= \mathbb E[C(t, T) {\mathbbm{1}_{t\leq \tau < T}} | \mathcal{F}_{t}] +
    \mathbb E[C(t, T) {\mathbbm{1}_{\tau \geq T}} | \mathcal{F}_{t}] -
    \mathbb E[(1 - R) (V_{\tau})^{+} {\mathbbm{1}_{t\leq \tau < T}} | \mathcal{F}_{t}] \\
    &=  \mathbb E[({\mathbbm{1}_{t\leq \tau < T}} + {\mathbbm{1}_{\tau \geq T}})C(t, T) | \mathcal{F}_{t}] -
    \mathbb E[(1 - R) (V_{\tau})^{+} {\mathbbm{1}_{t\leq \tau < T}} | \mathcal{F}_{t}] \\
    &= \mathbb E[C(t, T) | \mathcal{F}_{t}] -
    \mathbb E[(1 - R) (V_{\tau})^{+} {\mathbbm{1}_{t\leq \tau < T}} | \mathcal{F}_{t}] \\
    &= V_{t} - \mathbb E[(1 - R) (V_{\tau})^{+} {\mathbbm{1}_{t\leq \tau < T}} | \mathcal{F}_{t}].
\end{align*}
Furthermore, note that $(V_{\tau})^{+} = 0$ if $\tau \geq T$. Assuming that no default happened until $t$, we hence obtain 
\begin{align*}
    \CVA(t) = \mathbb E[(1 - R)(V_{\tau})^{+} | \mathcal{F}_{t}].
\end{align*}
Assume deterministic Loss Given Defaults, as well as the independence of credit risk to market factors affecting the BSM pricing, such that $\mathcal F_t = \mathcal G_t \cup \mathcal H_t$, where $\mathcal{G}_t$ comes from the randomness in the BSM model and $\mathcal{H}_t$ comes from the randomness in the default time.
Specifically, $V_t$ is adapted to the filtration $\mathcal G_t \subseteq \mathcal F_t$. 
Then, we separate these two sources of randomness:
\begin{align}
    \CVA(t) &= \mathbb E\left[ \mathbb E\left[(1 - R)\left(V_{\tau}\right)^{+} | \mathcal{G}_{t}\right] | \mathcal{H}_{t}\right] = \mathbb E[ g_{\mathcal{G}_{t}}(V,\tau) | \mathcal{H}_{t}] \ ,
\end{align}
where we define $g_{\mathcal{G}_{t}}(V,\tau) := \mathbb E\left[(1 - R)\left(V_{\tau}\right)^{+} | \mathcal{G}_{t}\right]$.
Transforming it into the integral form, we obtain
\begin{align*}
    \CVA(t) &= \int^T_{t} f_{\tau| \mathcal{H}_{t}}(u) g_{\mathcal{G}_{t}}(V,u) \dd{u}
\end{align*}
where $f_{\tau| \mathcal{H}_{t}}(u)$ is the probability measure on the default time.
Then,
\begin{align}
    \CVA(t) &= \int^T_{t} f_{\tau| \mathcal{H}_{t}}(u) \mathbb E[(1 - R)(V_{u})^{+} | \mathcal{G}_{t}] \dd{u} \\
    &= (1 - R)\int^T_{t} f_{\tau| \mathcal{H}_{t}}(u) \mathbb E[(V_{u})^{+} | \mathcal{G}_{t}] \dd{u} \\
    &\mathop{=}^{\lim_{n\to \infty}} (1 - R)\sum^{n - 1}_{i=0} \phi(\tau \in (t_i, t_{i + 1}]) \mathbb E[(V_{t_i})^{+} | \mathcal{G}_{t}] \ ,
\end{align}
where $\phi(\tau \in (t_i, t_{i + 1}]) = \Phi(\tau > t_i) - \Phi(\tau > t_{i+1})$ and $\Phi(\tau > t)$ is the survival probability up to time $t$ of the counterparty.

\section{Problems in Quantitative Finance} \label{append_quant}

\subsection{Introduction to Credit Curve Bootstrapping} \label{cds_bootsrapping}
The CVA formula was observed to be a linear combination in weights of expected exposure profiles, with the weights being defined by probability distributions of the default time $\tau$. In practice, these survival probabilities are taken from historical data or derived from implied Credit Default Swap (CDS) spreads using the risk-neutral measure. \cite{Green2015} 

We outline the process for obtaining default probabilities from market data, and refer interested readers to a more detailed treatment by Castellacci \cite{giuseppe}. We assume that the credit market for the derivatives in our portfolio are liquid; that CDS are readily traded and their market data is known. We operate under the settings of the `JPMorgan model', for which we outline the assumptions below.

Let the survival probability for default time be defined by $S(t) = 1 - \phi(t)$. The hazard rate corresponding to $\tau$ is defined via the deterministic function \cite{Green2015}:
\be
S(t) = \exp\left\{- \int ^t_0 h(u) du \right\}. \label{survival}
\ee 
By extension of the deterministic hazard rate process, it follows that hazard rates are independent of the other market variables under discussion, such as the discount factor. The credit default swap is a financial derivative that allows market participants to offset credit risk. The buyer of a CDS makes payments to the seller until some maturity date $T$. In return, the seller agrees that in the event that the reference entity defaults, the seller has to payout a sum as a percentage of the insured notional value. 

These payments that the buyer of a CDS pays is in the form of a spread, which is a percentage of the notional value paid out to the seller per annum. The trade value is characterized by two different `legs', known as the floating and the fixed leg. In a liquid market, the observables are these spreads at different maturities, and it is this curve representing spreads as a function of time that is coined as the term structure. Heuristically, we might expect that higher spreads are coincident with higher probabilities of default, since the rational seller/insurer shall demand higher premiums on insuring more risky reference securities. The default probabilities are said to be `implied' by the observed term structure of spreads. Our objective is to estimate these survival probabilities implied, where maturities of different length imply different risk expectations with respect to time.

Under the assumptions that the hazard rate is piecewise constant, we may partition the time axis up to maturity such that $0 = T_0 \leq T_1 \cdots \leq T_n = T$, where for $t \in [T_{i-1}, T_i)$, $h(t) = h_i \in \mathbb R$.

The survival probabilities are expressed \cite{giuseppe}:
\be
S(t) = \exp\left\{-\sum^{n(t)}_{i = 1} h_i \Delta T_i + h_{n(t) + 1}(t - T_{n(t)})\right\},
\ee
where $n(t) = \max\{i \leq n: T_i \leq t\}$ and $\Delta T_i = T_i - T_{i - 1}$.

The JPMorgan model makes certain key assumptions outlined:

\begin{enumerate}[i]
    \item Hazard rates are piecewise constant between different maturities.
    \item Default process is independent of the interest rate process.
    \item Default leg pays at the end of each accrual period.
    \item Occurrence of default is midway during each payment period.
    \item Accrual payment is made at the end of each period.
\end{enumerate}

Under these assumptions, we may attempt to value the legs at $T_0$ maturing at $T$. The present value of the fixed leg can be denoted
\be
V_{\rm fix}(T) = s\sum^n_{i = 1} \alpha_i Z(0, T_i) \left(S(T_i) + \frac{1}{2} (S(T_{i - 1}) - S(T_i)) \right),
\ee
and the present value of the floating leg to be
\be
V_{\rm float}(T) = (1 - R) \sum^n_{i = 1} Z(0, T_i) \left( S(T_{i - 1}) - S(T_i) \right ),
\ee
where $\alpha_i$ is the day count fraction between premium dates corresponding to the period between $T_{i - 1}$ and $T_i$ of a chosen convention, and $Z(0, T_i)$ is the value of a risk-free zero coupon bond starting from $T_0$ and maturing at $T_i$. If we assume a constant risk-free interest rate $r \in (0, 1)$, then $Z(0, T_i) = \exp \{- r T_i \}$. Otherwise, we may calibrate it to the interest rate term structure.

Note that a fair contractual agreement between the buyer and the seller at inception is such that $V_{\rm fix}(t) = V_{\rm float}(t)$ for all $t$. Note that the value of the contract at the onset must be zero, since the agreement is fair. In particular, let the value of the CDS contracts at time $t$ be denoted $C(t)$, then for all $t$, $V_{\rm float}(t) - V_{\rm fix}(t) = 0$. Assuming that the market data/spreads obtained are $s_1, s_2 \cdots s_m$ corresponding to maturities $T_{n_1}, T_{n_2} \cdots T_{n_m}$, we have $0 \overset{!}{=} C(T_{n_1}) = $
\be 
(1 - R) \sum^{n_1}_{i = 1} Z(0, T_i) \left( S(T_{i - 1}) - S(T_i) \right ) - s_1\sum^{n_1}_{i = 1} \alpha_i Z(0, T_i) \left(S(T_i) + \frac{1}{2} (S(T_{i - 1}) - S(T_i)) \right). 
\label{bootstrapping1}
\ee

Note that in the equation above, $R$ is a deterministic, known recovery rate, $Z(0, T_i)$ can be derived via calibrating the interest rate term structure, and $\alpha_i$ is also known. The only unknown variable is the survival probabilities which are a function of the hazard rates. In fact, the \cref{bootstrapping1} is an implicit equation on $h_1$ and may be solved via numerical solvers.

Repeating for the next maturity, we have $0 \overset{!}{=} C(T_{n_2})  = $
\be 
(1 - R) \sum^{n_2}_{i = 1} Z(0, T_i) \left( S(T_{i - 1}) - S(T_i) \right ) - s_1\sum^{n_2}_{i = 1} \alpha_i Z(0, T_i) \left(S(T_i) + \frac{1}{2} (S(T_{i - 1}) - S(T_i)) \right),
\label{bootstrapping2}
\ee
which is an implicit equation in $h_1$ and $h_2$. Using the result from approximation of $h_1$, we may use equivalent methods to derive $h_2$. These steps may be iterated up to $T_{n_m}$, and the survival probabilities may be determined via \cref{survival} vis-a-vis the hazard rates. The cumulative distribution function $\phi (t)$ immediately follows by taking the complement of $S(t)$.

\subsection{Discount Processes and the Interest Rate Term Structure} \label{irterms}
In our discussion, we have assumed access to oracles that give discounted values of the portfolio. In the Financial Preliminary, we introduced the concept of discounting, which is necessary to obtain the present value of future valuations of a portfolio. In particular, for the undiscounted portfolio value $V'_t$, the discounted one is $V_t = \exp \{ -\int^t_0 r_u du \}  V'_t$. Note that there exists short rate models that allow rates to take negative values, as is often observed empirically. A simple modelling choice would be to choose a deterministic discount factor $r_t = r$ calibrated to historical data. We may also model interest rates as stochastic processes evolving over time. Broadly speaking, interest rate models fall into the 4 categories - short rate models, Heath-Jarrow-Morton (HJM) models, Market Models and Markov Functional Models. In the context of XVA calculations, we are most concerned with the efficiency of their computation within the Monte Carlo simulation; and hence we prefer Markovian models over non-Markovian interest rate models in practise 
\cite{Green2015}. Markovian models have the advantage that they may be pre-computed in the initialization stage and then cached for use when the Monte Carlo paths are generated.
We give an overview of the discussion of such a model, called the Extended-Vasicek model under the settings of a Heath-Jarrow-Morton (HJM) framework.

The interest rate term structure may be described via the forward rates $f(t, T, T+\Delta)$, which is the interest rate for borrowing agreed to at time $t$, to borrow from time $T$ to $T + \Delta$. The instantaneous forward rate is the forward rate as the limit of $\Delta \to 0$, and we denote it as $f(t, T)$. The market observables relating to the term structure are these instantaneous forward rates at different starting times $t_i$ with different durations $T - t_i$. The short rates $r(t)$ may be defined in terms of the instantaneous forward rates; it is the rate agreed to borrow when $T - t_i$ goes to zero, such that $r(t) = f(t, t)$.

Here we introduce the concept of the zero coupon bond (ZCB), which is an asset that pays a dollar at maturity $T$ and pays no coupons. Denote the value of such an instrument evaluated at $t$ maturing at $T$ as $Z(t, T)$. By definition, $Z(T, T) = 1$. Under the principle of no arbitrage, $Z(t, T) = \exp \{-y(T - T) (T - t)\}$, where $y(T - t)$ is the interest yield from $t$ to $T$. Since we shall be indifferent between agreeing to borrowing at discrete intervals between $t$ and $T$ and agreeing to borrow in one contract, the value of the ZCB can be expressed in terms of the forward rates. That is:
\be
Z(t, T) &=& \exp \left\{ - \sum^{n - 1}_{i = 0} f(t, t_i) \Delta t_i \right\}
\\ &\overset{\lim_{n\to \infty}}{=}& \exp \left\{- \int ^T_t f(t, u) du \right\}.
\ee
The opposite of the ZCB is the money market account, or the bank account. Previously, we argued that the bank account at time $t$ has value $B_t = B_0 \exp \{rt\}$, where $B_0$ is the initial sum invested. However, treating interest rates as a stochastic process, we may more generally express it using the short rates:
\be 
B_t = B_0 \exp \left\{ \int^t_0 r(u) du \right\}. \label{bank_account}
\ee

The First Fundamental Theorem of Asset Pricing states that discounted asset prices are martingales under the risk-neutral measure $\mathbbm Q$. In particular, the ZCB scaled by the bank account is a martingale. Using the definitions of the bank account in \cref{bank_account} and that $Z(T, T) = 1$, we have:
\be
\frac{Z(t, T)}{B_t} = \mathbb E_{\mathbb Q} \left[ \frac{Z(T, T)}{B_T} | \mathcal F_t \right]
\to Z(t, T) = \mathbb E_{\mathbb Q} \left[ \exp\left\{ -\int^T_t r(u) du \right\} \ | \mathcal F_t \right].
\ee
Note that by definition, $Z(t, T + \Delta) = Z(t, T) \exp \{ -f(t, T, T + \Delta) \Delta \}$. We may express the instantaneous forward rates in terms of the ZCB \cite{brigo2006}:
\be
f(t, T) = - \frac{\partial}{\partial T} \ln Z(t, T). \label{zcb_dynamics} 
\ee
The dynamics of the instantaneous forward rate under HJM framework are stated \cite{Green2015}:
\be 
\dd f(t, T) = \mu(t, T) \dd t + \sigma_f(t, T) \dd W(t)
\ee
where $\mu$ is stochastic drift as a function of time, $\sigma_f$ the volatility of the instantaneous forward rate, and $W(t)$ is a Brownian motion.

To satisfy the principle of no arbitrage, it can be shown that the stochastic drift has the constraints such that the dynamics of the instantaneous forward rate $f$ under risk-neutral measure $\mathbbm Q$ is \cite{brigo2006}:
\be
\dd f(t, T) = \left( \sigma_f(t, T) \int^T_t \sigma_f(t, u) \dd u \right) \dd t + \sigma_f (t, T) \dd W_t.
\ee

An interest rate model that falls under the HJM framework is the Extended-Vasicek model \cite{Green2015}, which allows for fitting of the initial term structure by allowing the long term mean reversion level to  be a function of time. In particular, it models of the dynamics of the short rate as such \cite{brigo2006}:
\be
\dd r_t = \alpha(\theta_t - r_t) \dd t + \sigma \dd W_t,
\ee
where $\theta_t$ is the long term mean of the short rates as a function of time, and $\alpha$ is the speed of mean-reversion.

Relating it to the HJM framework and using stochastic calculus techniques \cite{Shreve2004}, it can be shown that the short rate dynamics are equivalent to \cite{Green2015}:
\be
\dd r_t = \alpha \left[ \frac{1}{\alpha} \frac{\partial f(0, t)}{\partial t} + \frac{\sigma_f ^2}{2\alpha ^2} (1 - \exp \{-2 \alpha t \}) + f(0, t) - r_t \right ] \dd t + \sigma_f \dd W_t \label{drift_const}
\ee

The calibration reduces to solving for $\theta_t = \frac{1}{\alpha} \frac{\partial f(0, t)}{\partial t} + \frac{\sigma_f ^2}{2\alpha ^2} (1 - \exp \{-2 \alpha t \}) + f(0, t)$ under a deterministic mean-reversion speed $\alpha$, known $f(0, t)$ and calibrating to ZCB prices implicitly related under the \cref{zcb_dynamics} using analytical or numerical solvers.

Similarly, the volatility process may be modelled as a stochastic function of time, and calibrated to interest rate derivatives such as a strip of co-terminal European swaptions. This is not, however, needed to fit the initial term structure.

\end{document}